\documentclass[10pt,journal,epsfig]{IEEEtran}

\usepackage[dvips]{graphicx}
\usepackage{graphicx}
\usepackage{amssymb}
\usepackage{cite}
\usepackage{subfigure}
\usepackage{amsmath}
\usepackage{algorithm}
\usepackage{algorithmic}
\usepackage{multirow}
\usepackage{color}

\usepackage{tikz}
\usepackage{lipsum}

\begin{document}

\title{Fast Beam Training and Alignment for IRS-Assisted Millimeter Wave/Terahertz Systems}

\author{Peilan Wang, Jun Fang, Wei Zhang ~\IEEEmembership{Fellow,~IEEE} and Hongbin Li,
~\IEEEmembership{Fellow,~IEEE}
\thanks{Peilan Wang, and Jun Fang are with the National Key Laboratory
of Science and Technology on Communications, University of
Electronic Science and Technology of China, Chengdu 611731, China,
Email: peilan\_wangle@std.uestc.edu.cn, JunFang@uestc.edu.cn}
\thanks{Wei Zhang is with the School of Electrical Engineering and Telecommunications,
The University of New South Wales, Sydney, Australia, E-mail: wzhang@ee.unsw.edu.au
}
\thanks{Hongbin Li is
with the Department of Electrical and Computer Engineering,
Stevens Institute of Technology, Hoboken, NJ 07030, USA, E-mail:
Hongbin.Li@stevens.edu}
\thanks{This work was supported in part by the National Science
Foundation of China under Grant 61829103.}
\thanks{\footnotesize  \textcopyright 2021 IEEE. Personal use of this material is permitted.
  Permission from IEEE must be obtained for all other uses, in any current or future
  media, including reprinting/republishing this material for advertising or promotional
  purposes, creating new collective works, for resale or redistribution to servers or
  lists, or reuse of any copyrighted component of this work in other works.}}

\maketitle



\begin{abstract}
Intelligent reflecting surface (IRS) has emerged as a competitive
solution to address blockage issues in millimeter wave (mmWave)
and Terahertz (THz) communications due to its capability of
reshaping wireless transmission environments. Nevertheless,
obtaining the channel state information of IRS-assisted systems is
quite challenging because of the passive characteristics of the
IRS. In this paper, we develop an efficient downlink beam
training/alignment method for IRS-assisted mmWave/THz systems.
Specifically, by exploiting the inherent sparse structure of the
base station-IRS-user cascade channel, the beam training problem
is formulated as a joint sparse sensing and phaseless estimation
problem, which involves devising a sparse sensing matrix and
developing an efficient estimation algorithm to identify the best
beam alignment from compressive phaseless measurements.
Theoretical analysis reveals that the proposed method can identify
the best alignment with only a modest amount of training overhead.
Numerical results show that, for both line-of-sight (LOS) and NLOS
scenarios, the proposed method obtains a significant performance
improvement over existing state-of-the-art methods. Notably, it
can achieve performance close to that of the exhaustive beam
search scheme, while reducing the training overhead by 95\%.
\end{abstract}


\section{Introduction}
Intelligent reflecting surface (IRS) has emerged as a competitive
solution to address blockage issues and extend the coverage in
millimeter wave (mmWave) and Terahertz (THz) communications
\cite{ZhangBjornson20,TanSun18,WangFang19,StefanRenzo19}. To reap
the gain brought by the large number of passive elements,
instantaneous channel state information (CSI) is required for
joint active and passive beamforming for IRS-assisted systems
\cite{WuZhang19a,WangFang20a,YuXu19,NingChen20,ZhangZhang20,MishraJohansson19,JensenDe20,GuanWu21}.
Nevertheless, CSI acquisition is challenging for IRS-assisted
mmWave/THz systems due to the passive characteristics of the IRS
and the large size of the channel matrix. Recently, some studies
proposed to utilize the inherent sparse/low-rank structure of the
cascade base station(BS)-IRS-user mmWave channel, and cast channel
estimation into a compressed sensing framework
\cite{WangFang20,ChenLiang19,HeYuan19,LiuGao20}. The proposed
methods, however, suffer several drawbacks.
\textcolor{black}{Firstly, sparse/low-rank signal recovery via
optimization methods or other heuristic methods usually incurs a
high computational complexity which might be excessive for
practical systems.} Secondly, compressed sensing methods require
accurate phase information of the received measurements. While in
mmWave bands, the carrier frequency offset (CFO) effect and the
random phase noise are more significant than that in sub-6GHz
bands \cite{MyersHeath17,MyersMezghani18,MyersHeath19}. As a
result, the phase of the measurements might be corrupted and
unavailable for channel estimation. \textcolor{black}{In
\cite{VanNgo21}, an aggregated channel estimation approach was
proposed for IRS-assisted cell-free massive MIMO systems, where
the reflecting coefficients of the IRS are pre-configured and thus
only the equivalent channel between the access point and the user
(also referred to as aggregated channel) needs to be estimated.
The aggregated channel can then be estimated via traditional
channel estimation methods. Nevertheless, this approach needs to
pre-configure IRS's reflecting coefficients based on statistical
CSI, thus may suffer a beamforming gain loss as compared with
those beamforming approaches that are based on instantaneous CSI.}




To address the above difficulties, instead of obtaining the full
CSI, we focus on the problem of beam training whose objective is
to acquire the angle of departure (AoD) and the angle of arrival
(AoA) associated with the dominant path between the BS and the
user. Beam training/alignment is an important topic that has been
extensively investigated in conventional mmWave systems, e.g.
\cite{AlkhateebEl14,XiaoHe16,NohZoltowski17,HassaniehAbari28,MyersMezghani18,WuCheng19,LiFang19}.
Nevertheless, beam alignment for IRS-assisted systems is more
challenging as we need to simultaneously align the BS-IRS link as
well as the IRS-user link. So far there are only a few attempts
made on beam alignment for IRS-assisted mmWave systems, e.g.
\cite{NingChen19,YouZheng20,WangZhang20,HuZhong20}. Specifically,
\cite{NingChen19} proposed a hierarchical beam search scheme for
IRS-assisted THz systems. The proposed scheme, however, requires
the BS to interact with each user individually, which may not be
feasible at the initial channel acquisition stage. In
\cite{YouZheng20}, a multi-beam sweeping method based on
grouping-and-extracting was proposed for beam training for
IRS-assisted mmWave systems. This work assumes that the BS has
aligned its beam to the LOS component between the BS and the IRS,
and then focuses on the beam training between the IRS and the
user. Nevertheless, in practice, the location information of the
IRS may not be available to the BS, in which case one needs to
perform a joint BS-IRS-user beam training. Recently,
\cite{WangZhang20} proposed a random beamforming-based maximum
likelihood (ML) estimation method to estimate the parameters
associated with the LOS component, by treating other NLOS
components as interference. In \cite{HuZhong20}, an uplink beam
training scheme was proposed for IRS-based multi-antenna multicast
systems, where the IRS is used as a component of the transmitter.





In this paper, we propose an efficient beam training scheme for
IRS-assisted mmWave/THz downlink systems. The key idea of our
proposed method is to let both the transmitter and the IRS form
multiple narrow beams to scan the angular space. Specifically, the
proposed beam training process consists of a few rounds of
``full-coverage scanning'', where in each full-coverage scanning,
the entire space is efficiently scanned using pre-designed
multi-directional beam training sequences. Also, in different
rounds of full-coverage scanning, we use different combinations of
directions to scan the angular space. Such a diversity allows us
to identify the best beam alignment via an efficient
set-intersection-based scheme. Theoretical analysis suggests that
the proposed method can identify the best alignment with only a
modest amount of training overhead.




It should be noted that although both the current work and
\cite{WangFang21} employ multi-directional beam sequences for
downlink training, there are some major distinctions between these
two works. \textcolor{black}{Specifically, the work
\cite{WangFang21} considered downlink training in OFDM systems,
where a same directional beam can be scaled by different factors
at different subcarriers. This feature was utilized such that
different directional beams can be distinguished from each other,
and thus the CSI can be conveniently extracted. Nevertheless, such
a scheme fails when single-carrier systems are considered because
the modulation vector degenerates into a scalar which is no longer
an effective fingerprint to distinguish different directional
beams. In contrast, our proposed method does not rely on any
modulation vector to identify the correct beam and works for
single-carrier systems.}



The rest of the paper is organized as follows. In Section
\ref{sec:system}, the system model and the problem formulation are
discussed. In Section \ref{sec:sensing}, we study how to devise
the active and passive beam training sequences. In Section
\ref{sec:ba-LOS} and Section \ref{sec:ba-NLOS}, we develop
efficient set-intersection-based methods for LOS and NLOS
scenarios to identify the best alignment based on compressive
phaseless measurements. Simulation results are provided in Section
\ref{sec:simulation}, followed by concluding remarks in Section
\ref{sec:conclusion}.


\begin{figure}[t]
\centering
\includegraphics[width=3.5in] {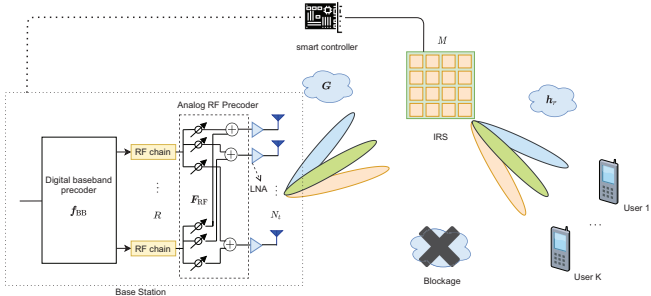}
\caption{IRS-assisted mmWave downlink multi-user systems.}
\label{Fig_system}
\end{figure}



\section{System Model and Problem Formulation} \label{sec:system}
\subsection{Downlink Training and System Model}
We consider the problem of downlink training and beam alignment in
an IRS-assisted mmWave/THz multi-user system, where an IRS is
deployed to assist the data transmission from the BS to a number
of single-antenna users (see Fig. \ref{Fig_system}).


Before proceeding, we first provide a discussion of our proposed
downlink training protocol. In the downlink beam training stage,
the BS periodically broadcasts a pre-designed common beam training
sequence $\{\boldsymbol{f}(t)\}_{t=1}^T$, and at the same time the
IRS uses a common reflection beam sequence
$\{\boldsymbol{v}(t)\}_{t=1}^T$ to reflect the signal coming from
the BS. For simplicity, we assume that the direct link from the BS
to each user is blocked due to unfavorable conditions, and the
transmitted signal arrives at each user via the reflected
BS-IRS-user channel\footnote{When the direct link between the BS
and the user is available, we can first switch off the IRS and
perform downlink training of the direct link using conventional
mmWave beam training schemes. The effect of the direct link can
then be canceled when we perform downlink BS-IRS-user training.}.
Each user receives signals reflected from the IRS, and estimates
the angular parameters associated with the dominant path of its
own downlink channel. This information is then sent by this user
to the BS via a dedicated channel. A connection can thus be
established between the BS and each user after the BS receives the
related channel information associated with this user. The
schematic of the downlink training protocol is also illustrated in
Fig. \ref{fig_protocol}.

\textcolor{black}{In this downlink training framework, we are interested in studying
\emph{how to jointly devise the active/passive beam training
sequences and estimate the channel parameters} to achieve fast
beam alignment between the BS and each user. Note that since
channel estimation is performed at each user and the design of
beam training sequences is independent of users, our study can be
simplified to single-user scenarios. Therefore in the rest of the
paper, we consider the scenario where there is only a single user
in the system. Also, in our setup, we assume that the BS has no
knowledge of the geographical location of all facilities,
including the IRS, the user and the BS itself.}


To more rigorously formulate our problem, we now proceed to
discuss our system model. The BS is equipped with $N_t$ antennas
and $R \ll N_t$ radio frequency (RF) chains. At the BS, a digital
baseband precoder $\boldsymbol{f}_{\text{BB}} \in \mathbb{C}^{R}$
is first applied to a broadcast signal ${s}$, then followed by an
analog RF beamformer $\boldsymbol{F}_{\text{RF}} \in
\mathbb{C}^{N_{t} \times R}$. The transmitted training signal at
the $t$th time instant can be written as
\begin{align}
\boldsymbol{x}(t) = \boldsymbol{F}_{\text{RF}}(t)
\boldsymbol{f}_{\text{BB}}(t)s(t)=\boldsymbol{f}(t)s(t),
\end{align}
where $\boldsymbol{f}(t)\triangleq \boldsymbol{F}_{\text{RF}}(t)
\boldsymbol{f}_{\text{BB}}(t)$ is the transmitter's beamforming
vector, and we set $s(t)=1$ in the beam training stage. Let
$\boldsymbol{G} \in \mathbb C^{M \times N_t}$ denote the channel
from the BS to the IRS, and $\boldsymbol{h}_{r} \in \mathbb C^{M}$
denote the channel from the IRS to the user. The IRS is a planar
array consisting of $M = M_y \times M_z$ passive reflecting
elements. Each reflecting element of the IRS can reflect the
incident signal with a reconfigurable phase shift and amplitude
via a smart controller. \textcolor{black}{Let $\phi_m (t)\in [0,2\pi]$ and $\zeta_m
(t)\in [0,1]$ denote the phase shift and amplitude coefficients
adopted by the $m$th element at the $t$th time instant.} Define
\begin{align}
    \boldsymbol{\Phi} (t)\triangleq {\rm diag}(\zeta_1 (t)e^{j\phi_1(t)},
    \cdots, \zeta_M(t) e^{j \phi_M(t)}) \in \mathbb C^{M \times M}.
\end{align}
as the IRS reflecting matrix. The training signal received by the
user can thus be expressed as
\begin{align}
z(t) =& \boldsymbol{h}_r^H \boldsymbol{\Phi} (t)\boldsymbol{G} \boldsymbol{x}(t) + n(t) \nonumber \\
=& \boldsymbol{v}^H(t) {\rm diag}(\boldsymbol{h}_r^H)
\boldsymbol{G} \boldsymbol{f}(t) + n(t) \nonumber \\
= & \boldsymbol{v}^H(t) \boldsymbol{H} \boldsymbol{f}(t) + n(t)
\label{received-signal-model}
\end{align}
where $\boldsymbol{v}(t)\triangleq [\zeta_1
(t)e^{j\phi_1(t)}\phantom{0} \cdots\phantom{0} \zeta_M(t) e^{j
\phi_M(t)}]^H$ is the passive reflecting vector, $\boldsymbol{H}
\triangleq {\rm
diag}(\boldsymbol{h}_{r}^H)\boldsymbol{G}\in\mathbb{C}^{M\times
N_t}$ is referred to as the cascade channel, and $n(t) \sim
\mathcal{CN}(0,\sigma^2)$ represents the additive white Gaussian
noise.


\begin{figure}[!t]
    \centering
    {\includegraphics[width=3.5in]{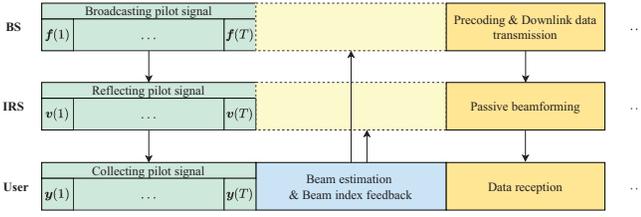}}
    \caption{Frame structure for the proposed joint BS-IRS-user beam training protocol.} \label{fig_protocol}
\end{figure}


\subsection{Channel Model}
It is well-known that the narrowband mmWave channel can be
characterized by a widely-used Saleh-Valenzuela (S-V) geometric
model \cite{AyachRajagopal14,JiangZhang19}. \textcolor{black}{As
for THz channels, some initial channel measurements at THz
frequencies \cite{HanBicen15,LinLi15,HanJornet18,TatariaShafi21}
reported that THz channels also exhibit sparse scattering
characteristics, an effect that is observable at mmWave
frequencies. The main difference between the mmWave and THz
channels lies in the path loss. Specifically, THz channels suffer
more severe free spreading loss due to the extremely high
frequency. In addition, the high attenuation caused by molecular
absorption is no longer negligible and has to be taken into
account at THz frequencies. On the other hand, due to the severe
path loss of THz frequencies, THz channels may exhibit a higher
degree of sparsity than mmWave communications. As a result, THz
channels can also be characterized by the S-V model.}
Specifically, the BS-IRS channel can be modeled as
\begin{align}
\boldsymbol{G} = &\sqrt{\frac{N_tM}{P }}
\big(\varrho_0\boldsymbol{a}_r(\vartheta _0^r,\gamma _0^r)
\boldsymbol{a}_{t}^H(\phi _0^t) \nonumber \\&+   \sum _{p=1}^{P-1}\varrho
_p\boldsymbol{a}_r(\vartheta _p^r,\gamma _p^r)
\boldsymbol{a}_{t}^H(\phi _p^t) \big) \label{ch-G}
\end{align}
where $P$ is the total number of paths between the BS and the IRS,
$\varrho_0$ denotes the complex gain of the LOS path,
\textcolor{black}{$\varrho_p,\forall p = 1, \ldots,P-1$
represents the complex gain of the $p$th NLOS path,
$\vartheta_p^r$ ($\gamma_p^r$) for $p=0,1, \ldots, P-1$ denotes the associated azimuth
(elevation) AoA, $\phi_p^t$ for $p=0, \ldots,P-1$ is the associated AoD, and
$\boldsymbol{a}_r$ ($\boldsymbol{a}_{t}$) denotes the normalized
receive (transmit) array response vector.} For simplicity, we
define
\begin{align}
\boldsymbol{a} (\phi,N) \triangleq \frac{1}{\sqrt{N}}[1\phantom{0}
e^{j\pi\phi}\phantom{0}\ldots\phantom{0} e^{j \pi(N-1) \phi}]^T.
\end{align}
Suppose the BS employs a uniform linear array (ULA). The transmit
array response vector can be expressed as
\begin{align}
\boldsymbol{a}_t(\phi) = \boldsymbol{a}\bigg(\frac{2d}{\lambda}
\sin \phi, N_t\bigg).
\end{align}
Also, since the IRS is an $M_y \times M_z$ uniform planar array,
the receive array response vector can be written as
\begin{align}
\boldsymbol{a}_r (\vartheta,\gamma) =
\boldsymbol{a}\bigg(\frac{2d}{\lambda} \sin \vartheta\sin \gamma,
M_y\bigg) \otimes \boldsymbol{a}\bigg(\frac{2d}{\lambda} \cos
\gamma, M_z\bigg).
\end{align}
where $\otimes$ denotes the Kronecker-product, $d$ is the antenna
spacing and $\lambda$ denotes the wavelength of the signal.

Owing to the sparse scattering nature of mmWave channels, the
BS-IRS channel has a sparse representation in the angular
(beam-space) domain:
\begin{align}
\boldsymbol{G} =& \sqrt{\frac{N_tM}{P}}(\boldsymbol{D}_{M_y}
\otimes \boldsymbol{D}_{M_z}) \boldsymbol{\Sigma}
\boldsymbol{D}_{N_t}^H = \sqrt{\frac{N_tM}{P}} \boldsymbol{D}_{\rm
R} \boldsymbol{\Sigma} \boldsymbol{D}_{N_t}^H \label{eqn3}
\end{align}
where $\boldsymbol{D}_{\rm R}\triangleq \boldsymbol{D}_{M_y}
\otimes \boldsymbol{D}_{M_z} \in \mathbb C^{M \times M}$,
$\boldsymbol{D}_{M_y}$, $\boldsymbol{D}_{M_z}$, and
$\boldsymbol{D}_{N_t}$ are defined as
\begin{align}
\boldsymbol{D}_{N} \triangleq [ \boldsymbol{a}(\eta(1),N), \ldots,
\boldsymbol{a}(\eta(N),N)] \in \mathbb C^{N \times N}, \label{D-N}
\end{align}
with $N=M_y,M_z,N_t$ respectively, $\eta(i) = -1+ \frac{2i-1}{N},
i = 1,\ldots,N$, and $\boldsymbol{\Sigma}\in \mathbb C^{M\times
N_t}$ is a sparse matrix with $P$ nonzero entries. Here we suppose
that the true AoA and AoD lie on the discretized grid. In the
presence of grid mismatch, the number of nonzero elements in the
sparse matrix will increase as a result of power leakage.

Similarly, the IRS-user channel can be modeled as
\begin{align}
\boldsymbol{h}_r =\sqrt{\frac{M}{P' }} \bigg( \alpha _0
\boldsymbol{a}_r(\vartheta _0^t,\gamma _0^t)+\sum
_{p=1}^{P^{\prime }-1} \alpha _p \boldsymbol{a}_r(\vartheta
_p^t,\gamma _p^t)  \bigg), \label{hr}
\end{align}
where $P'$ is the number of signal paths between the IRS and the
user, $\alpha_0$ denotes the complex gain of the LOS path,
$\alpha_p,\forall p=1,\ldots,P'-1$ denotes the complex gain
associated with the $p$th NLOS path, and $\vartheta_p^t$
($\gamma_p^t$) for $p=0,1,\ldots,P'-1$ denotes the associated
azimuth (elevation) AoD. Due to sparse scattering characteristics,
the IRS-user channel can be written as
\begin{align}
\boldsymbol{h}_{r} = \sqrt{\frac{M}{P'}}\boldsymbol{ D}_{\rm R}
\boldsymbol{\alpha}, \label{eqn4}
\end{align}
where $\boldsymbol{\alpha} \in \mathbb C^{M \times 1}$ is a sparse
vector with $P^{\prime}$ nonzero entries. It is easy to verify
that $\boldsymbol{D}_{\rm R}^H\boldsymbol{ D}_{\rm R} =
\boldsymbol{I}_{M}$ and $\boldsymbol{D}_{N_t}^H
\boldsymbol{D}_{N_t} = \boldsymbol{I}_{N_t}$.


Based on (\ref{eqn3}) and (\ref{eqn4}), it was shown in
\cite{WangFang20} that the cascade channel admits a sparse
representation as follows
\begin{align}
\boldsymbol{H} =& \sqrt{\frac{N_tM}{PP'}} \boldsymbol{\bar D}_R
\boldsymbol{\tilde \Lambda} \boldsymbol{D}_{N_t}^H \nonumber \\
=&\boldsymbol{\bar D}_R \boldsymbol{\Lambda}
\boldsymbol{D}_{N_t}^H, \label{eqn1}
\end{align}
where $\boldsymbol{\bar D}_R \triangleq \boldsymbol{\tilde
D}_R(:,1:M)$ is a submatrix of $\boldsymbol{\tilde D}_R$
constructed by its first $M$ columns, and $\boldsymbol{\tilde D}_R
\triangleq \sqrt{M} \boldsymbol{D}_{\rm R}^{\ast} \bullet
\boldsymbol{D}_{\rm R}$, with $\bullet$ denoting the transposed
Khatri-Rao product. It can be readily verified that
$\boldsymbol{\tilde D}_R\in\mathbb{C}^{M\times M^2}$ has $M$
distinct columns which are exactly its first $M$ columns. Also, we
can verify that $\boldsymbol{\bar D}_R$ is an unitary matrix, i.e.
$\boldsymbol{\bar D}_R^{H}\boldsymbol{\bar
D}_R=\boldsymbol{I}_{M}$, and its column takes the form of
$\boldsymbol{\bar a}_r( \varphi,\varpi) \triangleq
\boldsymbol{a}(\varphi, M_y) \otimes \boldsymbol{a}(\varpi,M_z)$.

Also, we have
\begin{align}
\boldsymbol{\Lambda} \triangleq \sqrt{\frac{N_tM}{PP'}}
\boldsymbol{\tilde \Lambda}\in \mathbb C^{ M \times N_t}
\end{align}
in which $\boldsymbol{\tilde \Lambda}$ is a merged version of
$\boldsymbol{J}\triangleq \boldsymbol{\alpha}^{\ast} \otimes
\boldsymbol{\Sigma}$, with each of its rows being a superposition
of a subset of rows in $\boldsymbol{J}$, i.e.
\begin{align}
\boldsymbol{\tilde \Lambda}(i,:) = \sum_{n \in \mathcal{S}_{i}}
\boldsymbol{J}(n,:)
\end{align}
where $\boldsymbol{\tilde \Lambda}(i,:)$ denotes the $i$th row of
$\boldsymbol{\tilde \Lambda}$, $\mathcal{S}_{i}$ denotes the set
of indices associated with those columns in $\boldsymbol{\tilde
D}_R$ that are identical to the $i$th column of $\boldsymbol{\bar
D}_R$. It is clear that $\boldsymbol{\Lambda}$ is a sparse matrix
with at most $P\times P'$ nonzero elements.

\subsection{Problem Formulation}
Combining (\ref{received-signal-model}) and (\ref{eqn1}), the
received pilot signal at the user can be expressed as
\begin{align}
z(t)=&\boldsymbol{v}^H(t) \boldsymbol{H} \boldsymbol{f}(t) +
n(t) \nonumber\\
=&\boldsymbol{v}^H(t)\boldsymbol{\bar D}_R\boldsymbol{\Lambda}
\boldsymbol{D}_{N_t}^H\boldsymbol{f}(t) + n(t) \label{eqn2}
\end{align}
Note that (\ref{eqn2}) is an ideal signal model without
considering the CFO and random phase noise. In mmWave/THz bands,
the CFO, i.e., the mismatch in the carrier frequencies at the
transmitter and the receiver, is more significant than sub-6GHz
bands and cannot be neglected. For instance, a small offset of 10
parts per million (ppm) at mmWave frequencies can cause a large
phase misalignment in less than a hundred nanoseconds. Besides the
CFO, mmWave communication systems also suffer random phase noise
due to the jitter of the oscillators. The phase noise, together
with the CFO, leads to an unknown phase shift to measurements
$z(t)$ that varies across time. In this case, only the magnitude
of the measurement $z(t)$ is reliable. Define
\begin{align}
y(t)\triangleq |z(t)|= |\boldsymbol{v}^H(t)\boldsymbol{\bar
D}_R\boldsymbol{\Lambda} \boldsymbol{D}_{N_t}^H\boldsymbol{f}(t) +
n(t)| \label{phaseless-signal}
\end{align}

Based on $\{y(t)\}_{t=1}^T$, our objective is to acquire the
information needed to achieve beam alignment between the BS and
the user. Note that identifying the best beam alignment is
equivalent to acquiring the location index of the largest (in
magnitude) element in the sparse matrix $\boldsymbol{\Lambda}$.
This is because $\boldsymbol{\Lambda}$ is a beam-space
representation of the cascade channel $\boldsymbol{H}$. Hence the
largest element in $\boldsymbol{\Lambda}$ actually corresponds to
the strongest path of the BS-IRS-user channel.


To identify the largest element in $\boldsymbol{\Lambda}$, a
natural approach is to exhaustively search all possible beam
pairs. Specifically, at each time instant $t$, choose
$\boldsymbol{f}(t)\triangleq\boldsymbol{F}_{\text{RF}}(t)
\boldsymbol{f}_{\text{BB}}(t)$ as a certain column from
$\boldsymbol{D}_{N_t}$, and choose $\boldsymbol{v}(t)$ as a
certain column from $\boldsymbol{\bar D}_R$, i.e.
\begin{align}
\boldsymbol{v}(t) = & \sqrt{M} \boldsymbol{\bar D}_{R} (:,i) , \\
\boldsymbol{f}(t) = & \boldsymbol{D}_{N_t} (:,j),
\end{align}
Here the scaler $\sqrt{M}$ in $\boldsymbol{v}(t)$ is used to
ensure that entries of $\boldsymbol{v}(t)$ are of constant
modulus. Then the received measurement $y(t)$ is given by
\begin{align}
y(t) = | \sqrt{M} \lambda_{i,j} +n(t)|
\end{align}
where $\lambda_{i,j}$ denotes the $(i,j)$th entry of
$\boldsymbol{\Lambda}$. After an exhaustive search, we can
identify the largest (in magnitude) entry in
$\boldsymbol{\Lambda}$. This exhaustive search scheme, however,
has a sample complexity of $M N_t$, which is prohibitively high
since both $M$ and $N_t$ are large for mmWave and THz systems in
order to combat severe path loss. In the following sections, we
develop a more efficient method to perform joint BS-IRS-user beam
training.

Specifically, since the period of time for beam training is
proportional to the number of measurements $T$, the problem of
interest is how to devise the active/passive beam training
sequences $\{\boldsymbol{f}(t),\boldsymbol{v}(t)\}_{t=1}^T$ and
develop a computationally efficient estimation scheme such that we
can identify the best beam alignment using as few measurements as
possible.

\emph{Remark:} Different from our work that neglects the phase
information of the received measurements, we noticed that some
other works, e.g. \cite{MyersMezghani18,MyersHeath17}, model the
CFO as an unknown parameter and perform joint CFO and channel
estimation.








\section{Beam Training Sequence Design} \label{sec:sensing}
To more efficiently probe the channel, we propose to let the BS
and the IRS form multiple pencil beams simultaneously and steer
them towards different directions. Specifically, the precoding
vector $\boldsymbol{f}(t)$ is chosen to be
\begin{align}
\boldsymbol{f}(t) = & \boldsymbol{F}_{\rm RF}(t)
\boldsymbol{f}_{\rm BB}(t)
= \boldsymbol{D}_{N_t} \boldsymbol{S}(t) \boldsymbol{f}_{\rm BB}(t) =\boldsymbol{D}_{N_t}
\boldsymbol{a}(t), \label{hyb-f}
\end{align}
where $\boldsymbol{S}(t) \in \{0,1\}^{N_t \times R}$ is a column
selection matrix which has only one nonzero element in each column
and $\boldsymbol{a}(t) \triangleq \boldsymbol{S}(t)
\boldsymbol{f}_{\rm BB}(t)$ is a sparse vector with at most $R$
nonzero entries. Note that each column of $\boldsymbol{D}_{N_t}$
can be considered as a beamforming vector steering a beam to a
certain direction. Hence, the hybrid precoding vector in
\eqref{hyb-f} can form at most $R$ beams towards different
directions simultaneously. The passive reflecting vector
$\boldsymbol{v}(t)$ can be generated in a similar way. We let
\begin{align}
\boldsymbol{v}(t) = \boldsymbol{\bar D}_R \boldsymbol{c}(t),
\label{pass-v}
\end{align}
where $\boldsymbol{c}(t)$ is a sparse vector containing at most
$Q$ nonzero elements. Here $Q$ is a parameter of user's choice. We
will discuss its choice later in this paper.


Substituting \eqref{hyb-f}-\eqref{pass-v} into
\eqref{phaseless-signal}, we obtain
\begin{align}
{y} (t) =  | \boldsymbol{c}^H(t)
\boldsymbol{\Lambda}\boldsymbol{a}(t) + n(t)| \label{y-ca}
\end{align}

\subsection{Sensing Matrix Design}
In this subsection, we discuss how to devise a set of sparse
sensing vectors $\{\boldsymbol{c}(t),\boldsymbol{a}(t)\}_{t=1}^T$
to efficiently probe the channel. Let $S_{c}(t) \triangleq \{ i|
c_{i}(t) >0\}$ denote the set of indices associated with the
nonzero elements in $\boldsymbol{c}(t)$, and $S_{a}(t) \triangleq
\{ j| a_{j}(t)
>0\}$ denote the set of indices of the nonzero elements in
$\boldsymbol{a}(t)$. Also, for simplicity, we assume that the
nonzero entries in $\{\boldsymbol{c}(t)\}$ are all set to
$\beta>0$, and the nonzero entries in $\{\boldsymbol{a}(t)\}$ are
all set to $\gamma>0$. Therefore, we have
\begin{align}
{y} (t) =&\big|\beta\gamma\sum_{i\in S_c(t),j\in
S_a(t)}\lambda_{i,j}+n(t)\big| \label{eqn9}
\end{align}


To find the strongest signal path, we need to make sure that each
element in $\boldsymbol{\Lambda}$ will be scanned at least once.
We first introduce the concept of ``a round of full-coverage
scanning'' as a basic building block for our beam training
process. Define $U\triangleq M/Q$ and $V\triangleq N_t/R$ and
assume both of them are integers. Each round of full-coverage
scanning consists of $T_0=U\times V$ measurements, and these $T_0$
measurements are generated according to:
\begin{align}
\boldsymbol{Y}=|\boldsymbol{C}^H\boldsymbol{\Lambda}\boldsymbol{A}+\boldsymbol{N}|
\end{align}
where $\boldsymbol{C}\in\mathbb{R}^{M\times U}$,
$\boldsymbol{A}\in\mathbb{R}^{N_t\times V}$, and
$\boldsymbol{Y}\in\mathbb{R}^{U\times V}$ is a matrix constructed
by $\{y(t)\}_{t=1}^{T_0}$. Specifically, the $(u,v)$th entry of
$\boldsymbol{Y}$ is equal to $\boldsymbol{Y}(u,v)=y((u-1)V+v)$,
and $(u,v)$th entry of $\boldsymbol{N}$ is equal to
$\boldsymbol{N}(u,v)=n((u-1)V+v)$. Therefore, once
$\boldsymbol{C}$ and $\boldsymbol{A}$ are specified, the set of
sparse sensing vectors
$\{\boldsymbol{c}(t),\boldsymbol{a}(t)\}_{t=1}^{T_0}$ can be
accordingly determined. Let
\begin{align}
\boldsymbol{C}=&[\boldsymbol{c}_1\phantom{0}\boldsymbol{c}_2\phantom{0}\ldots\phantom{0}
\boldsymbol{c}_U] \nonumber\\
\boldsymbol{A}=&[\boldsymbol{a}_1\phantom{0}\boldsymbol{a}_2\phantom{0}\ldots\phantom{0}
\boldsymbol{a}_V]
\end{align}
From the relation between $\boldsymbol{Y}$ and $\{y(t)\}$, it is
clear that we have
$\boldsymbol{c}_u=\boldsymbol{c}((u-1)V+v),\forall v$ and
$\boldsymbol{a}_v=\boldsymbol{a}((u-1)V+v),\forall u$. The set of
sparse encoding vectors $\{\boldsymbol{c}_u\}$ and
$\{\boldsymbol{a}_v\}$ are devised to satisfy the following two
conditions:
\begin{itemize}
\item[C1] Those nonzero entries in $\{\boldsymbol{c}_u\}$ and
$\{\boldsymbol{a}_v\}$ are respectively set to $\beta$ and
$\gamma$. Also, we have $\|\boldsymbol{c}_u\|_{0}=Q,\forall u$,
and $\|\boldsymbol{a}_v\|_{0}=R,\forall v$.
\item[C2] The sparse vectors in $\{\boldsymbol{c}_u\}$ are
orthogonal to each other, i.e.
$\boldsymbol{c}_{u_1}^T\boldsymbol{c}_{u_2}=0, \forall u_1\neq
u_2$; and vectors in $\{\boldsymbol{a}_v\}$ are orthogonal to each
other, i.e. $\boldsymbol{a}_{v_1}^T\boldsymbol{a}_{v_2}=0, \forall
v_1\neq v_2$.
\end{itemize}
Let $S(t)\triangleq\{\lambda_{i,j}\}_{i\in S_{c}(t),j\in
S_{a}(t)}$ denote the set of elements that are simultaneously
sensed/hashed at the $t$th time instant. Such a set of elements is
also called as a bin, as illustrated in Fig. \ref{fig_example}.
Clearly we have $|S(t)|=QR$. Also, condition C2 ensures that the
sets of elements sensed at different time instants are disjoint,
i.e.
\begin{align}
S(t_1)\cap S(t_2)=\emptyset, \forall t_1\neq t_2
\end{align}
In addition, since we have $UQ=M$ and $RV=N_t$, the union of the
sets is equal to the whole set of elements of
$\boldsymbol{\Lambda}$, i.e.
\begin{align}
S(1)\cup\ldots\cup S(T_0)=\{\lambda_{i,j}\}_{i=1,j=1}^{i=M,j=N_t}
\end{align}




After a single round of full-coverage scanning, no element in
$\boldsymbol{\Lambda}$ is left unscanned. Nevertheless, since each
element in $\boldsymbol{\Lambda}$ is scanned along with other
elements at each time, we still cannot identify the exact location
of the largest component from the measurements $\boldsymbol{Y}$.
To identify the strongest component, we need to perform a few
rounds, say $L$ rounds, of full-coverage scanning, and for each
round of scanning, we randomly generate $\boldsymbol{A}$ and
$\boldsymbol{C}$ by altering locations of the nonzero entries in
$\{\boldsymbol{c}_u\}$ and $\{\boldsymbol{a}_v\}$. We will show
that we can identify the largest element in $\boldsymbol{\Lambda}$
via a simple decoding scheme from these $L$ rounds of measurements
$\{\boldsymbol{Y}_{l}\}_{l=1}^L$. Here $\boldsymbol{Y}_{l}$
denotes the measurement matrix collected at the $l$th round of
scanning, and we have
\begin{align}
\boldsymbol{Y}_l=|\boldsymbol{C}_l^H\boldsymbol{\Lambda}\boldsymbol{A}_l+\boldsymbol{N}_l|
\label{eqn5}
\end{align}
where $\boldsymbol{C}_l\triangleq
[\boldsymbol{c}^{l}_{1}\phantom{0}\ldots\phantom{0}\boldsymbol{c}^{l}_{U}]$
and $\boldsymbol{A}_l\triangleq
[\boldsymbol{a}^{l}_{1}\phantom{0}\ldots\phantom{0}\boldsymbol{a}^{l}_{V}]$
are sparse encoding matrices used in the $l$th round of scanning.

\subsection{Practical Considerations of Devising $\{\boldsymbol{c}_u\}$}
As discussed in the previous subsection, the vectors
$\{\boldsymbol{c}_u\}$ are devised to be strictly sparse with $Q$
nonzero elements. To fulfill this requirement, we need to have an
independent control of the reflection amplitude for each IRS
element, which increases not only the hardware complexity but also
the energy consumption \cite{AbeywickramaZhang20,TangChen21}. Moreover, to generate a
strictly sparse vector $\boldsymbol{c}_u$, many of the reflection
amplitudes have to be set far less than one, which reduces the
reflection efficiency.


To cope with these issues, we wish to find a set of passive
beamforming vectors $\{\boldsymbol{v}_u\}$ with constant modulus,
and the corresponding vectors $\{\boldsymbol{c}_u=\boldsymbol{\bar
D}_R^H\boldsymbol{v}_u\}$ are approximately-sparse vectors with
$Q$ dominant entries. Mathematically, this problem can be
formulated as follows. Given any $Q$ columns from
$\boldsymbol{\bar D}_R$, denoted as
$\{\boldsymbol{p}_q\}_{q=1}^Q$, let $\boldsymbol{Q}$ be a matrix
constructed by these $Q$ columns and $\boldsymbol{\bar{Q}}$ be a
matrix obtained by removing those $Q$ columns from
$\boldsymbol{\bar D}_R$. We seek a constant-modulus vector
$\boldsymbol{v}$ such that $|\boldsymbol{v}^H\boldsymbol{Q}|$ is a
quasi-constant magnitude vector with its magnitude as large as
possible, whereas $\|\boldsymbol{v}^H\boldsymbol{\bar{Q}}\|_2$ is
as small as possible. There are different approaches to tackle
this problem. Inspired by \cite{WangFang20a}, here we formulate
the above problem into the following optimization:
\begin{align}
\max_{\boldsymbol{v}} \quad &
\sum_{q=1}^{Q} \log_2(1+ \boldsymbol{v}^H \boldsymbol{p}_{q} \boldsymbol{p}_{q}^H \boldsymbol{v})  \nonumber \\
{\text s.t.} \quad & |v_i| = 1, \forall i = 1,\ldots ,M
\label{v-opt-robust}
\end{align}
where $v_i$ denotes the $i$th entry of the vector
$\boldsymbol{v}$. It was shown in \cite{WangFang20a} that the
solution to \eqref{v-opt-robust} is nearly orthogonal to
$\boldsymbol{\bar{Q}}$. Moreover, entries of the vector
$\boldsymbol{v}^H\boldsymbol{Q}$ have quasi-constant magnitudes
thanks to the logarithmic function. As a result, the resulting
vector $\boldsymbol{c}=\boldsymbol{\bar D}_R^H\boldsymbol{v}$ is
an approximately sparse vector with $Q$ dominant entries. Note
that the generated vectors $\{\boldsymbol{c}_u\}$ cannot be
strictly orthogonal to each other since they are no longer
strictly sparse vectors. Nevertheless, for each round of
full-coverage scanning, it is not difficult to attain
near-orthogonality by making sure that the sets of dominant
elements sensed at different time instants are disjoint. Also, as
will be shown later in this paper, our proposed algorithm requires
the indices of those nonzero elements in $\{\boldsymbol{c}_u\}$ to
identify the best alignment. As $\{\boldsymbol{c}_u\}$ generated
from (\ref{v-opt-robust}) are approximately sparse, we only
consider these $Q$ prominent entries as nonzero elements of
$\boldsymbol{c}_u$.


The above optimization can be efficiently solved by a
manifold-based algorithm, which has a very low computational
complexity of $\mathcal{O}(M)$ \cite{WangFang20a}. Besides, the
reflecting vectors $\{\boldsymbol{v}_u\}$ can be calculated and
stored in advance. It will not exert an extra computational burden
on the beam alignment task.



\begin{figure*}[!t]
\centering
\includegraphics[width=6in] {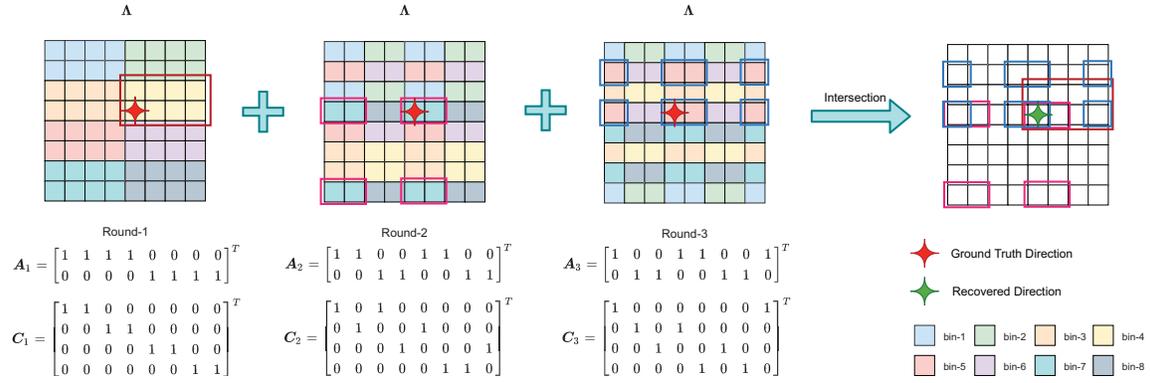}\hfill
\caption{An illustrative example to show how the proposed scheme
recovers the location of the largest element in
$\boldsymbol{\Lambda}$. A bin consists of those entries in
$\boldsymbol{\Lambda}$ that are simultaneously sensed at each time
instant. For different rounds, we randomize the entries that fall
into different bins. In each round, the bin associated with the
largest measurement $y(t)$ is highlighted with a rectangular box.
By performing the intersection operation, the entry associated
with the dominant path can be estimated as the common element of
those rectangular boxes.} \label{fig_example}
\end{figure*}

\section{Proposed Beam Alignment Method: LOS Scenarios} \label{sec:ba-LOS}
In the previous section, we have discussed how to devise the
active and passive beam training sequences
$\{\boldsymbol{f}(t)\}_{t=1}^T$ and
$\{\boldsymbol{v}(t)\}_{t=1}^T$. In this section, we discuss how
to identify the best beam alignment (i.e., identify the largest
element in $\boldsymbol{\Lambda}$) from the received phaseless
measurements $\{y(t)\}_{t=1}^T$ for the LOS scenario. Note that
this estimation task is performed at the receiver, i.e. user.


We first consider the scenario where there is only one nonzero
element or only one prominent nonzero element in the matrix
$\boldsymbol{\Lambda}$. This scenario has important practical
implications and arises as a result when both the BS-IRS channel
and the IRS-user channel are LOS-dominated. As reported in many
real-world channel measurements
\cite{AkdenizLiu14,Muhi-EldeenIvrissimtzis10}, the power of mmWave
LOS path is much higher (about 13 dB higher) than the sum of the
power of NLOS paths. When it comes to the THz bands, the power of
the LOS component is about 20dB higher than the power of the
scattering components \cite{PriebeKannicht13}. Therefore it can be
expected that $\boldsymbol{\Lambda}$ contains only one dominant
element when the LOS path is available for both the BS-IRS and the
IRS-user links.

To better illustrate the idea of the proposed scheme, we consider
a noiseless case where the measurements $\{\boldsymbol{Y}_l\}$ are
not corrupted by noise. When $\boldsymbol{\Lambda}$ contains only
one dominant element, it is clear that the measurement matrix
$\boldsymbol{Y}_l$ collected at the $l$th round of scanning
contains only one prominent component whose location can be easily
determined. Suppose that, for each $l$,
$\boldsymbol{Y}_l(u_l,v_l)$ is the largest element in
$\boldsymbol{Y}_l$. From (\ref{eqn5}), we have
\begin{align}
\boldsymbol{Y}_l(u_l,v_l)=(\boldsymbol{c}^{l}_{u_l})^H\boldsymbol{\Lambda}\boldsymbol{a}^{l}_{v_l}
\end{align}
Let $S_{u_l}^{(l)} \triangleq \{ i| \boldsymbol{c}^{l}_{u_l}(i)
>0\}$ denote the indices of the nonzero elements in
$\boldsymbol{c}^{l}_{u_l}$, and $S_{v_l}^{(l)} \triangleq \{ j|
\boldsymbol{a}^{l}_{v_l}(j) >0\}$ denote the indices of the
nonzero elements in $\boldsymbol{a}^{l}_{v_l}$.

Let $(i^{\ast},j^{\ast})$ denote the location index of the
dominant entry in $\boldsymbol{\Lambda}$. It is clear that we have
\begin{align}
i^{\ast} \in S_{u_l}^{(l)}, \quad j^{\ast} \in S_{v_l}^{(l)},
\quad \forall l
\end{align}
As a result, we have
\begin{align}
i^{\ast} \in  \bigcap_{l=1}^L S_{u_l}^{(l)} \qquad j^{\ast} \in
\bigcap_{l=1}^L S_{v_l}^{(l)} \label{opt-index}
\end{align}
On the other hand, since $\boldsymbol{c}^{l}_{u_l}$ and
$\boldsymbol{a}^{l}_{v_l}$ are randomly generated for each round
of scanning, it is unlikely that there exists another location
index $(i,j)$ which lies in the intersection of these sets,
particularly when $L$ is large. Therefore we can determine the
location index of the dominant entry, $(i^{\ast},j^{\ast})$, as
\begin{align}
i^{\ast} = \bigcap_{l=1}^L S_{u_l}^{(l)} \qquad j^{\ast}
=\bigcap_{l=1}^L S_{v_l}^{(l)} \label{intersection-scheme}
\end{align}
In Fig. \ref{fig_example}, we provide an illustrative example to
show how to identify the largest element via an intersection
scheme.

Our following theorem shows that such an intersection scheme can
recover the true location of the dominant element with a high
probability. The main results are summarized as follows.


\newtheorem{theorem}{Theorem}
\begin{theorem}
Suppose $Q < M$ and $R < N_t$. After $L$ rounds of full-coverage
scanning, from (\ref{intersection-scheme}), we can identify the
location of the nonzero element in $\boldsymbol{\Lambda}$ with a
probability greater than
\begin{align}
{\rm P}\geq  & \left(  1- (M-1) \left(\frac{Q-1}{M-1} \right)^L \right)  \times \left( 1- (N_t-1)
\left(\frac{R-1}{N_t-1} \right)^L
\right) \nonumber\\
=&P(Q,L,M) \times P(R,L,N_t) \label{eqn7}
\end{align}
where the function $P(x,y,z) $ is defined as
\begin{align}
P(x,y,z) \triangleq  1- (z-1) \left(\frac{x-1}{z-1} \right)^y
\end{align}
\label{theorem1}
\end{theorem}
\begin{proof}
See Appendix \ref{appendixA}.
\end{proof}

\subsection{Sample Complexity Analysis}
We now analyze the sample complexity of the proposed scheme. To
ensure that we can recover the index of the dominant element of
$\boldsymbol{\Lambda}$ with a probability exceeding a predefined
threshold $p_0$, we need
\begin{align}
{\rm P}\geq &  P(Q,L,M) \times P(R,L,N_t)   \geq p_0
\end{align}
For simplicity, we set $p_0 = p_1 \times p_2$, and let
\begin{align}
P(Q,L,M) \geq p_1  \label{P-Q-1} \\
P(R,L,N_t) \geq p_2 \label{P-R-2}
\end{align}
From \eqref{P-Q-1}, it is easy to verify that
\begin{align}
L \geq \frac{ \log \frac{1-p_1}{M-1}}{ \log \frac{Q-1}{M-1}} =
\frac{ \log ({M-1})+c_1 }{ \log \frac{M-1}{Q-1}} \triangleq L_1
\end{align}
where $c_1 \triangleq \log(({1-p_1})^{-1})  >0 $ is a constant. On
the other hand, from \eqref{P-R-2}, we have
\begin{align}
L \geq \frac{ \log \frac{1-p_2}{N_t-1}}{ \log \frac{R-1}{N_t-1}} =
\frac{ \log ({N_t-1})+ c_2}{ \log \frac{N_t-1}{R-1}} \triangleq
L_2
\end{align}
where $c_2 = \log(({1-p_2})^{-1}) $ is a constant.


Therefore, the total number of measurements $T$ required for
identifying the strongest component with a probability at least
$p_0$ can be calculated as
\begin{align}
T =  {UV}L \geq  UV \max \{ L_1,L_2\}
\end{align}
Since $L_1$ ($L_2$) is in the order of $\mathcal{O}(\log(M))$
($\mathcal{O}(\log(N_t))$), the proposed intersection-based scheme
has a sample complexity of $\mathcal{O}(UV\max\{\log(M),
\log(N_t)\})$. Recall that $U=M/Q$, where $Q$ is a parameter of
user's choice. Thus, we can choose a proper $Q$ to obtain a small
value of $U$. To be specific, given $T$ and other system
parameters, we can try different combinations of $(Q,L)$ to
determine the one which yields the highest probability of correct
beam alignment. Here we provide an example to show how many
measurements are exactly required to achieve perfect beam
alignment with a decent probability. Suppose $N_t=128$, $M=256$,
$R=4$. For different choice of $Q$ and $L$, our proposed method
can identify the best beam alignment with a probability no less
than:
\begin{itemize}
  \item $Q=32$, $L=4$, $T=UVL=1024$: ${\rm P}\geq 94.43\%$
  \item $Q=16$, $L=3$, $T=UVL=1536$: ${\rm P}\geq 94.65\%$
  \item $Q=16$, $L=4$, $T=UVL=2048$: ${\rm P}\geq 99.69\%$
\end{itemize}
From this example, we see that the proposed scheme can achieve a
substantial training overhead reduction as compared with the
exhaustive search scheme which requires a total number of
measurements up to $T=MN_t=32768$.


\subsection{Extension To The Noisy Case}
The proposed intersection scheme may not work well in the presence
of noise. In the sequel, inspired by \cite{HassaniehAbari28}, we
develop a noisy version of the intersection scheme. The basic idea
is to assign each element in $\boldsymbol{\Lambda}$ a probability
instead of a $0/1$-hard vote, and turn the intersection operation
into a product of probabilities.

Specifically, for each round of scanning and each index $(i,j)$ of
the element in $\boldsymbol{\Lambda}$, we define an indicator
matrix, $\boldsymbol{I}^{(l)}_{ (i,j)} \in \mathbb \{0,1\}^{U
\times V}$, with its $(u,v)$th entry defined as
\begin{align}
\boldsymbol{I}^{(l)}_{(i,j)} (u,v) \triangleq \begin{cases} 1&
\text{if $\boldsymbol{c}^l_{u}(i)\times
\boldsymbol{a}^l_{v}(j)\neq 0$} \\ 0&\text{otherwise}.
\end{cases} \label{cov-func}
\end{align}
where $\boldsymbol{x}(j)$ denotes the $j$th element of the vector
$\boldsymbol{x}$. Clearly, if the element
$\boldsymbol{\Lambda}(i,j)$ is sensed at the time instant $(
(u-1)V+v  )$ of the $l$th round, the value in \eqref{cov-func}
would be $1$; otherwise it would be $0$.

Based on the indicator matrix, we can further calculate the
``probability'' matrix $\boldsymbol{\mathcal{P}}^{(l)}$ with its
$(i,j)$th entry defined as
\begin{align}
\boldsymbol{\mathcal{P}}^{(l)} ({i,j})\triangleq \left( {\text
{vec}} (\boldsymbol{I}^{(l)}_{(i,j)} ) \right)^T {\text{vec}}
(\boldsymbol{Y}_l \circ \boldsymbol{Y}_l^{\ast}) \label{eqn6}
\end{align}
where $ {\text{vec}} (\cdot)$ denotes the vectorization operator
and $\circ$ represents the Hadamard product. Here (\ref{eqn6})
uses the received magnitude measurement as a weight to calculate
the probability of the element $\boldsymbol{\Lambda}(i,j)$ being a
dominant entry in $\boldsymbol{\Lambda}$.

Generally, if $\boldsymbol{\mathcal{P}}^{(l)} ({i,j})\geq
\epsilon$, where $\epsilon$ is a pre-specified threshold, then
$(i,j)$ is regarded as a candidate index of the dominant entry in
$\boldsymbol{\Lambda}$. Let $\mathcal{F}_l = \{ (i,j) |
\boldsymbol{\mathcal{P}}^{(l)} ({i,j}) \geq \epsilon \}$ denote
the set of candidate indices obtained from the $l$th round. After
$L$ rounds of scanning, we can determine the location index of the
dominant entry $(i^{\ast},j^{\ast})$ via a maximum likelihood (ML)
estimation
\begin{align}
(i^{\ast},j^{\ast}) = \max_{(i,j) \in \mathcal{F}} \quad\prod_{l =
1}^L \boldsymbol{\mathcal{P}}^{(l)} ({i,j}) \label{prob-tot}
\end{align}
where $\mathcal{F} \triangleq \cup_{l=1}^L \mathcal{F}_l$ denotes
the set comprising all candidate indices.
The overall algorithm for LOS scenarios are summarized in Algorithm \ref{Algorithm1}.

\begin{algorithm}[H]   
\caption{Proposed beam alignment algorithm for LOS scenarios}
\label{Algorithm1}
\begin{algorithmic}[1]
\STATE Generate $\{ \boldsymbol{A}_l \}_{l=1}^L$ according to
Section \ref{sec:sensing}.A; \STATE Generate $\{ \boldsymbol{C}_l
\}_{l=1}^L$ according to Section \ref{sec:sensing}.B; \STATE
Obtain received signals $\{\boldsymbol{Y}_l\}_{l=1}^L$ for $L$
rounds; \FOR{$l = 1,\ldots,L$} \STATE Calculate the ``probability"
function via \eqref{eqn6} for all indices $(i,j)$. \ENDFOR \STATE Determine the best beam
direction index $(i^{\ast},j^{\ast})$ of the dominant element in
$\boldsymbol{\Lambda}$ via \eqref{prob-tot}.
\end{algorithmic}
\end{algorithm}


\subsection{Computational Complexity Analysis}
The major computational task of our proposed beam estimation
method is to calculate the probability matrix defined in
\eqref{eqn6}. According to \eqref{eqn6}, each entry of the
probability matrix is calculated as an inner product of two
$UV$-dimensional vectors. Since each element in
$\boldsymbol{\Lambda}$ is sensed only once in each round, the
indicator matrix $\boldsymbol{I}^{(l)}_{(i,j)}$ contains only one
nonzero entry. Therefore each entry of the probability matrix can
be calculated by multiplying this nonzero entry with its
corresponding entry in ${\text{vec}} (\boldsymbol{Y}_l \circ
\boldsymbol{Y}_l^{\ast})$. As a result, calculating this entire
probability matrix $\boldsymbol{\mathcal{P}}^{(l)} \in \mathbb
R^{M \times N_t}$ involves a computational complexity of
$\mathcal{O}(M N_t)$. Note that our proposed method requires to
compute a set of probability matrices
$\{\boldsymbol{\mathcal{P}}^{(l)}\}_{l=1}^L$, which has a
computational complexity in the order of $\mathcal{O}(M N_t L)$,
where $L$ is the number of rounds of full-coverage scanning. After
obtaining $\{\boldsymbol{\mathcal{P}}^{(l)}\}_{l=1}^L$, we need to
calculate the objective function defined in \eqref{prob-tot},
which is a Hadamard product of the set of probability matrices
$\{\boldsymbol{\mathcal{P}}^{(l)}\}_{l=1}^L$ and involves a
computational complexity of $\mathcal{O}(MN_t (L-1))$. From the
above discussion, we see that the overall computational complexity
of our proposed method is in the order of $\mathcal{O}(M N_t L)$.

As a comparison, as analyzed in \cite{WangFang20}, if we employ a
compressed sensing-based method to recover the cascade channel,
the method needs to solve a sparse signal recovery problem of size
$T\times (M N_t)$, whose computational complexity is of
$\mathcal{O}(M N_t T \bar K)$ for greedy methods and of
$\mathcal{O}(M^3 N_t^3)$ for more sophisticated methods such as
the basis pursuit. Here $T$ denotes the number of measurements
used for channel estimation, and $\bar K$ is the sparsity level of
the cascade channel matrix. Our analysis shows that for our
proposed method, only a few rounds of full-coverage scanning (say,
$L<10$) are sufficient to identify the best beam alignment with a
probability close to $1$. Hence generally we have $L\ll T\bar K$.
Therefore, our proposed method has a much lower complexity even
compared with the least computationally demanding compressed
sensing method.


\section{Proposed Beam Alignment Method for NLOS Scenarios} \label{sec:ba-NLOS}
In this section, we extend our proposed method to a more general
case where there are multiple comparable nonzero elements in the
sparse matrix $\boldsymbol{\Lambda}$. Such a scenario arises as a
result of NLOS transmissions when either the BS-IRS's LOS path or
the IRS-user's LOS path is blocked by obstacles. For the case
where multiple paths of comparable qualities are available from
the BS to the user, the signals from different paths may combine
destructively at the receiver, thus creating difficulties to
identify the strongest path. Due to this destructive multi-path
effect, a direct application of the above proposed method may
result in beam misalignment. To address this issue, we propose a
modified version of the set intersection-based method for
identifying the strongest component in $\boldsymbol{\Lambda}$.

We first consider a noiseless case to illustrate the idea of the
proposed estimation scheme. Recalling \eqref{eqn9}, we have
\begin{align}
y(t)=\bigg|\beta\gamma\sum_{S(t)}\lambda_{i,j}\bigg|
\end{align}
where $S(t)=\{\lambda_{i,j}\}_{i\in S_c(t),j\in S_{a}(t)}$ denotes
the set of elements that are simultaneously sensed at the $t$th
time instant. According to the number of nonzero elements in the
set $S(t)$, the associated received signal $y(t)$ is called as a
nullton, a singleton, and a multiton if:
\begin{itemize}
\item Nullton: The received signal $y(t)$ is a nullton if
its associated set $S(t)$ contains no nonzero element.
\item Singleton: The received signal $y(t)$ is a singleton
if its associated set $S(t)$ includes only a single nonzero
element.
\item Multiton: The received signal $y(t)$ is a multiton if its
associated set $S(t)$ includes more than one nonzero elements.
\end{itemize}
Also, if the measurements collected within a certain round of
scanning, say $\boldsymbol{Y}_{l}$, only contain singleton and
nullton measurements, then this round of scanning is referred to
as a no-multiton (NM) round.

The basic idea of our proposed scheme is to utilize the
measurements associated with those NM rounds of scanning to
identify the largest component in $\boldsymbol{\Lambda}$. Since
the NM rounds consist of only singleton and nullton measurements,
it means that signals from different paths are separately sensed
and will not be hashed to contribute to a same measurement. Thus
the signals from different paths will not be combined
destructively at the receiver.

Nevertheless, we first need to differentiate NM rounds from those
rounds of scanning which include multiton measurements. Suppose
that the sparse matrix $\boldsymbol{\Lambda}$ contains $K$ nonzero
elements, where $K\ll MN_t$. Recall that for each round of
scanning, the sets of elements sensed at different time instants
are disjoint, i.e. $S(t_1)\cap S(t_2)=\emptyset,\forall t_1\neq
t_2$, and the union of the sets is the whole set of elements of
$\boldsymbol{\Lambda}$, i.e. $S(1)\cup\ldots\cup
S(T)=\{\lambda_{i,j}\}_{i=1,j=1}^{i=M,j=N_t}$. Therefore for an NM
round, it should contain exactly $K$ singleton measurements and
$UV-K$ nullton measurements. On the other hand, if a round of
scanning is not an NM round, then it should include more than
$UV-K$ nullton measurements because some of the $K$ nonzero
elements in $\boldsymbol{\Lambda}$ are sensed simultaneously.
Motivated by this observation, we can consider those rounds with
the smallest number of nulltons as NM rounds, without assuming the
knowledge of $K$. Note that determining whether $y(t)$ is a
nullton measurement or not is simple in the noiseless case because
we have $y(t)=0$ if $y(t)$ is a nullton. In the noisy case, an
energy detector can be employed to differentiate nulltons from
singletons and multitons.

After those NM rounds are identified, we can employ the
intersection-based scheme to find the strongest component in
$\boldsymbol{\Lambda}$. Suppose there are $\bar{L}$ NM rounds
among all $L$ rounds, and denote the set of NM rounds as
$\mathcal{L}$. Suppose that, for each $l\in\mathcal{L}$,
$\boldsymbol{Y}_l(u_l,v_l)$ is the largest (in magnitude) element
in $\boldsymbol{Y}_l$. From (\ref{eqn5}), we have
\begin{align}
\boldsymbol{Y}_l(u_l,v_l)=(\boldsymbol{c}^{l}_{u_l})^H\boldsymbol{\Lambda}\boldsymbol{a}^{l}_{v_l}
\end{align}
Let $S_{u_l}^{(l)} \triangleq \{ i| \boldsymbol{c}^{l}_{u_l}(i)
>0\}$ denote the indices of the nonzero elements in
$\boldsymbol{c}^{l}_{u_l}$, and $S_{v_l}^{(l)} \triangleq \{ j|
\boldsymbol{a}^{l}_{v_l}(j) >0\}$ denote the indices of the
nonzero elements in $\boldsymbol{a}^{l}_{v_l}$. Let
$(i^{\ast},j^{\ast})$ denote the location index of the largest (in
magnitude) component in $\boldsymbol{\Lambda}$. It is clear that
we have
\begin{align}
i^{\ast} \in S_{u_l}^{(l)}, \quad j^{\ast} \in S_{v_l}^{(l)},
\quad \forall l\in\mathcal{L}
\end{align}
As a result, we can estimate $(i^{\ast},j^{\ast})$ as
\begin{align}
i^{\ast} = \bigcap_{l \in  \mathcal{L} }S_{u_l}^{(l)} \qquad
j^{\ast} =\bigcap_{l\in \mathcal{L}}S_{v_l}^{(l)}
\label{intersection-scheme-2}
\end{align}



\subsection{Theoretical Analysis}
From the above discussion, we see that our proposed method relies
on those measurements collected within the NM rounds of scanning
to find the best beam alignment. A natural question is: how likely
a round of full-coverage scanning is an NM-round of scanning? We
have the following results regarding this question.
\newtheorem{proposition}{Proposition}
\begin{proposition}
Suppose the location indices of the $K$ nonzero components in
$\boldsymbol{\Lambda} \in \mathbb C^{M \times N_t}$ are uniformly
distributed. The sparse encoding matrices $\boldsymbol{C}_l \in
\{0,1 \} ^{M\times U}$ and $\boldsymbol{A}_l \in \{0,1\}^{N_t
\times V}$ for the $l$th round of scanning are designed to satisfy
conditions C1 and C2. Let $E$ denote the event that the $l$th
round of scanning is an NM round. Then we have
\begin{align}
P(E) = (RQ)^K \frac{\binom{UV}{K}}{\binom{MN_t}{K}} \triangleq p
,\label{NM-graph-probability}
\end{align}
where $R = N_t/V$, and $Q=M/U$. \label{proposition1}
\end{proposition}
\begin{proof}
This result is an extension of Proposition 1 in \cite{LiFang19}. The
proof is thus omitted here.  
\end{proof}

Here we provide an example to show the probability of a round of
full-coverage scanning being an NM-round of scanning. Suppose $N_t
= 128$, $M = 256$, and $R=4$. For different values of $Q$ and $K$,
we have
\begin{itemize}
\item $Q=64$, $K=4$: $p = 95.40\%$
\item $Q= 32$, $K=4$: $p=97.69\%$
\item $Q=32$, $K=2$: $p=99.61\%$
\end{itemize}
We see that with a reasonable choice of $Q$, the measurement
matrix $\boldsymbol{Y}_l$ is very likely to contain only singleton
and nullton measurements. Hence after a few rounds of scanning, it
can be expected that most of these rounds of scanning are NM
rounds.

Based on Proposition \ref{proposition1}, the probability with
which the set-intersection estimator (\ref{intersection-scheme-2})
can find the largest (in magnitude) element in
$\boldsymbol{\Lambda}$ can be characterized as follows.
\begin{theorem}
Suppose there are $1 <K \ll MN_t$ nonzero elements in the sparse
matrix $\boldsymbol{\Lambda}$. After $L$ rounds of full-coverage
scanning, from \eqref{intersection-scheme-2}, the proposed method
can identify the location of the largest element in
$\boldsymbol{\Lambda}$ with a probability ${\rm P}'$ that can be
bounded as
\begin{align}
{\rm P}' \geq \sum_{l = 0}^{L}  g(Q,l,M) \times g(R,l,N_t)    \times
\left( \binom{L}{l} p^{l}(1-p)^{L-l}  \right) \label{prob-2}
\end{align}
where $p$ is defined in \eqref{NM-graph-probability} and the
function $g(Q,l,M)$ is defined as
\begin{align}
 g(Q,l,M)  \triangleq 1- \sum_{j=1}^{Q-1} (-1)^{(j-1)} \binom{M-1}{j}
\bigg(\frac{\binom{M-1-j}{Q-1-j}}{\binom{M-1}{Q-1}} \bigg)^l.
\end{align}
\label{theorem2}
\end{theorem}
\begin{proof}
See Appendix \ref{appC}.
\end{proof}

Here we provide an example to show the probability \eqref{prob-2}
of identifying the largest component in $\boldsymbol{\Lambda}$.
Suppose $N_t = 128$, $M = 256$, and $R = 4$. For different values
of $Q$, $K$ and $L$, we have
\begin{itemize}
\item $Q=32$, $K=4$, $L=4$, $T=  1024$: ${\rm P}'\geq  91.52\%$
\item $Q= 32$, $K=4$, $L=5$, $T = 1280$: ${\rm P}' \geq 98.63\%$
\item $Q=16$, $K=2$, $L=4$, $T = 2048$: ${\rm P}' \geq 99.65\%$
\end{itemize}
Compared this example with the one in Section \ref{sec:ba-LOS}.B,
we can see that the proposed method for the NLOS scenario can
achieve a decent probability of exact recovery with a sample
complexity similar to (or slightly higher than) that of the LOS
scenario.

\subsection{Extension To The Noisy Case}
\label{LOS-decode-robust} We now extend our proposed estimation
method to the noisy case. When the measurements are corrupted by
noise, the received signal can be expressed as
\begin{align}
\boldsymbol{Y}_l(u,v)=
|(\boldsymbol{c}^{l}_{u})^H\boldsymbol{\Lambda}\boldsymbol{a}^{l}_{v}
+ n_{uv}^l| \quad \forall u,v
\end{align}
where $n_{uv}^l \sim \mathcal{CN}(0,\sigma^2)$ denotes the
additive noise. We first need to determine whether
$\boldsymbol{Y}_l (u,v)$ is a nullton measurement or not. Such a
problem can be formulated as a binary hypothesis test problem:
\begin{align}
  &H_0: \boldsymbol{Y}_l(u,v) = |{n}_{uv}^{l}| ,\nonumber \\
  &H_1: \boldsymbol{Y}_l(u,v)=   \bigg|\beta\gamma\sum_{i\in S_{u}^{(l)},j\in
S_{v}^{(l)}}\lambda_{i,j}+{n}_{uv}^{l}\bigg|
\end{align}
where $S_{u}^{(l)} \triangleq \{ i| \boldsymbol{c}^{l}_{u}(i)
>0\}$ denote the indices of the nonzero elements in
$\boldsymbol{c}^{l}_{u}$, and $S_{v}^{(l)} \triangleq \{ j|
\boldsymbol{a}^{l}_{v}(j) >0\}$ denote the indices of the
nonzero elements in $\boldsymbol{a}^{l}_{v}$. A simple energy detector can
be used to perform the detection
\begin{align}
\boldsymbol{Y}_l(u,v)\underset{H_0}{\overset{H_1}{\gtrless}}\epsilon
\label{energy-detect}
\end{align}
Given a specified false alarm probability, the threshold
$\epsilon$ can be easily determined since $\boldsymbol{Y}_l(u,v)$
follows a Rayleigh distribution under $H_0$. Since the received
signals are corrupted by noise, the selection of $\epsilon$ can
result in different performance. To harness the advantage of
multiple full-coverage scanning rounds, we often set $\epsilon$ to
be a small value. Next, we choose those rounds of scanning with
the least number of nulltons as NM rounds. Specifically, let
$\mathcal{L}$ with $|\mathcal{L}|=\bar{L}$ denote the set of NM
rounds. We can utilize the robust scheme developed in Section
\ref{sec:ba-LOS}.C to estimate the location of the largest entry
in $\boldsymbol{\Lambda}$, i.e.,
\begin{align}
(i^{\ast},j^{\ast}) = \max_{(i,j) \in \mathcal{F}} \quad\prod_{l
\in \mathcal{L} }
 \boldsymbol{\mathcal{P}}^{(l)} ({i,j}) \label{prob-tot-2}
\end{align}
where $ \boldsymbol{\mathcal{P}}^{(l)} ({i,j})$ is defined in
\eqref{eqn6}.

For clarity, the algorithm is summarized in Algorithm
\ref{Algorithm2}. Following a similar analysis in Section
\ref{sec:ba-LOS}.C, we know that the proposed scheme has a
computational complexity of order $\mathcal{O}(M N_t \bar{L})$,
where $\bar{L}$ denotes the number of NM rounds.

\begin{algorithm}[H]
\caption{Proposed beam alignment algorithm for NLOS scenarios}
\label{Algorithm2}
\begin{algorithmic}[1]
\STATE Generate $\{ \boldsymbol{A}_l \}_{l=1}^L$ according to
Section \ref{sec:sensing}.A; \STATE Generate $\{ \boldsymbol{C}_l
\}_{l=1}^L$ according to Section \ref{sec:sensing}.B; \STATE
Obtain received signals $\{\boldsymbol{Y}_l\}_{l=1}^L$ for $L$
rounds; \FOR{$l=1,\ldots,L$} \STATE Determine whether the received
signal ${Y}_l(u_l,v_l)$ is a nullton or not via the energy
detector \eqref{energy-detect}. Count the number of nulltons for
each round. \ENDFOR \STATE Find rounds with the smallest number of
nulltons and regard them as NM rounds. \FOR{$l = 1,\ldots,L$}
\IF{the round-$l$ is an NM round} \STATE Calculate the
``probability" function via \eqref{eqn6} for all indices $(i,j)$.
\ENDIF \ENDFOR \STATE Determine the best beam direction index
$(i^{\ast},j^{\ast})$ of the dominant element in
$\boldsymbol{\Lambda}$ via \eqref{prob-tot-2}.
\end{algorithmic}
\end{algorithm}



\begin{figure*}[!t]
\centering
 \subfigure[Success rate versus $T$ for LOS scenarios.]
 {\includegraphics[width=2.8in]{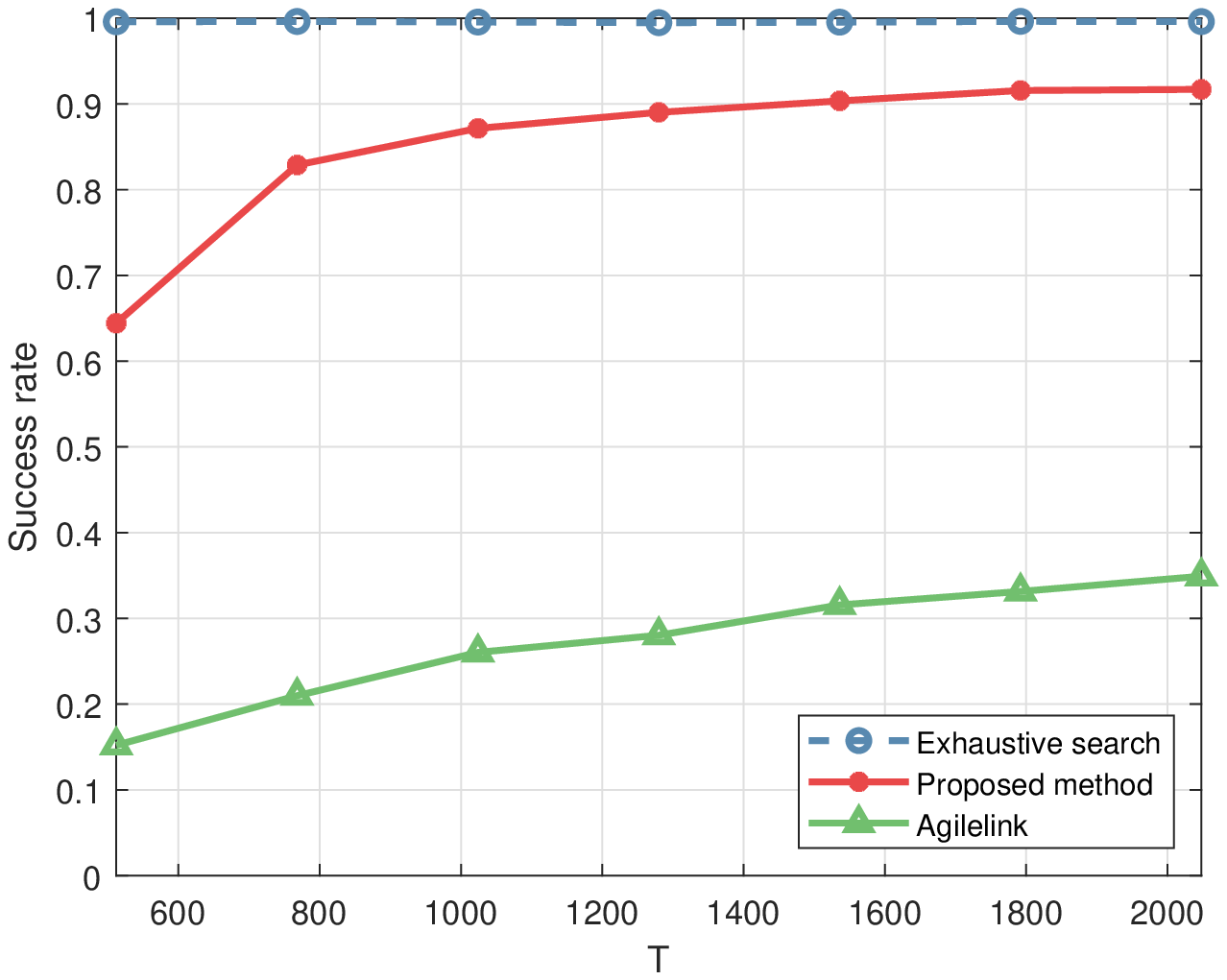}} \hfil
\subfigure[BGR versus $T$ for LOS
scenarios.]{\includegraphics[width=2.8in]{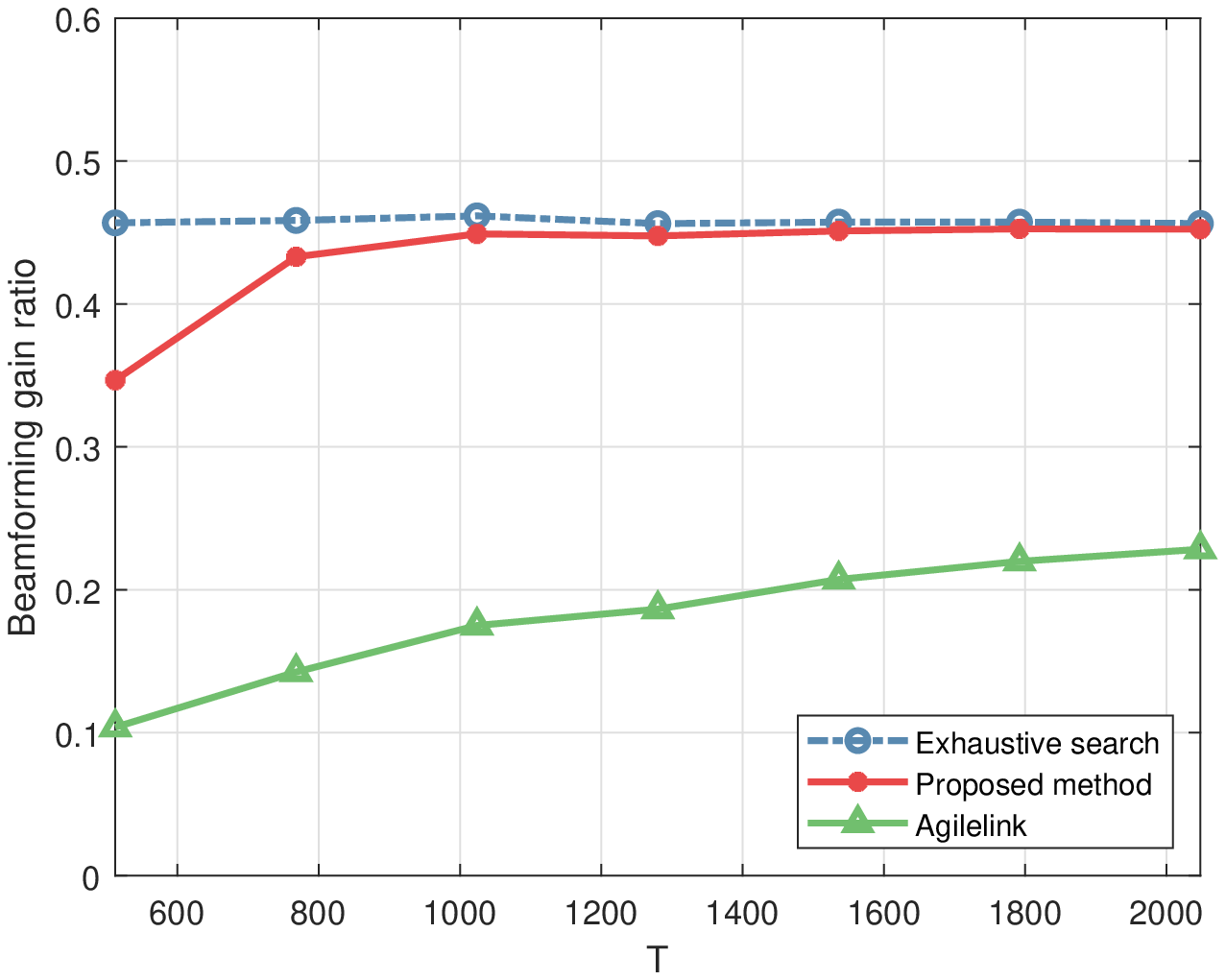}}\hfil
 \subfigure[Success rate versus $T$ for NLOS scenarios.]
 {\includegraphics[width=2.8in]{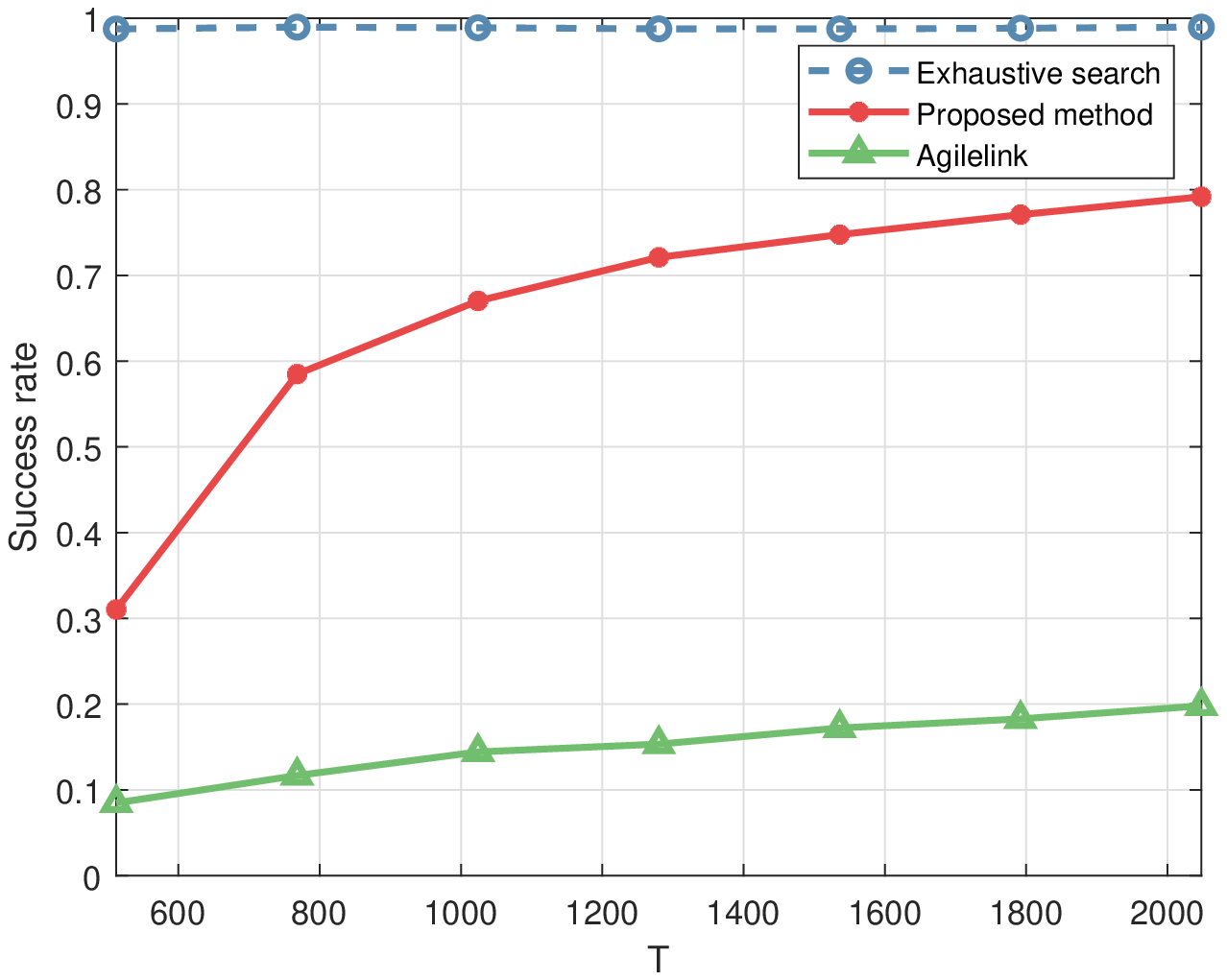}} \hfil
\subfigure[BGR versus $T$ for NLOS
scenarios.]{\includegraphics[width=2.8in]{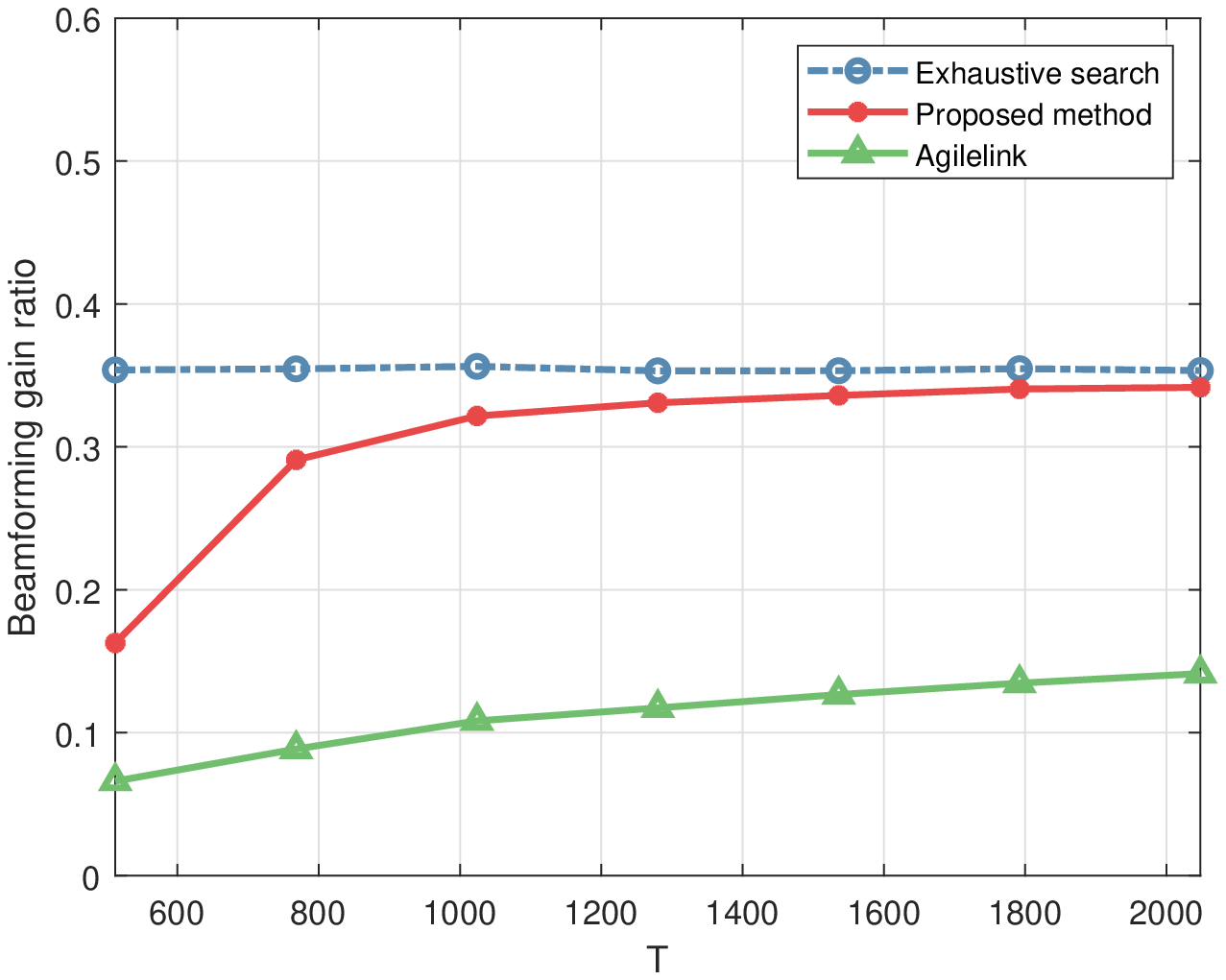}}
\caption{Success rate and beamforming gain ratio versus $T$ for
LOS and NLOS scenarios in the noisy case, $\text{SNR} = -20$dB.} \label{fig_L_large}
\end{figure*}


\begin{figure*}[!t]
\centering
 \subfigure[Success rate versus SNR for LOS scenarios.]
 {\includegraphics[width=2.8in]{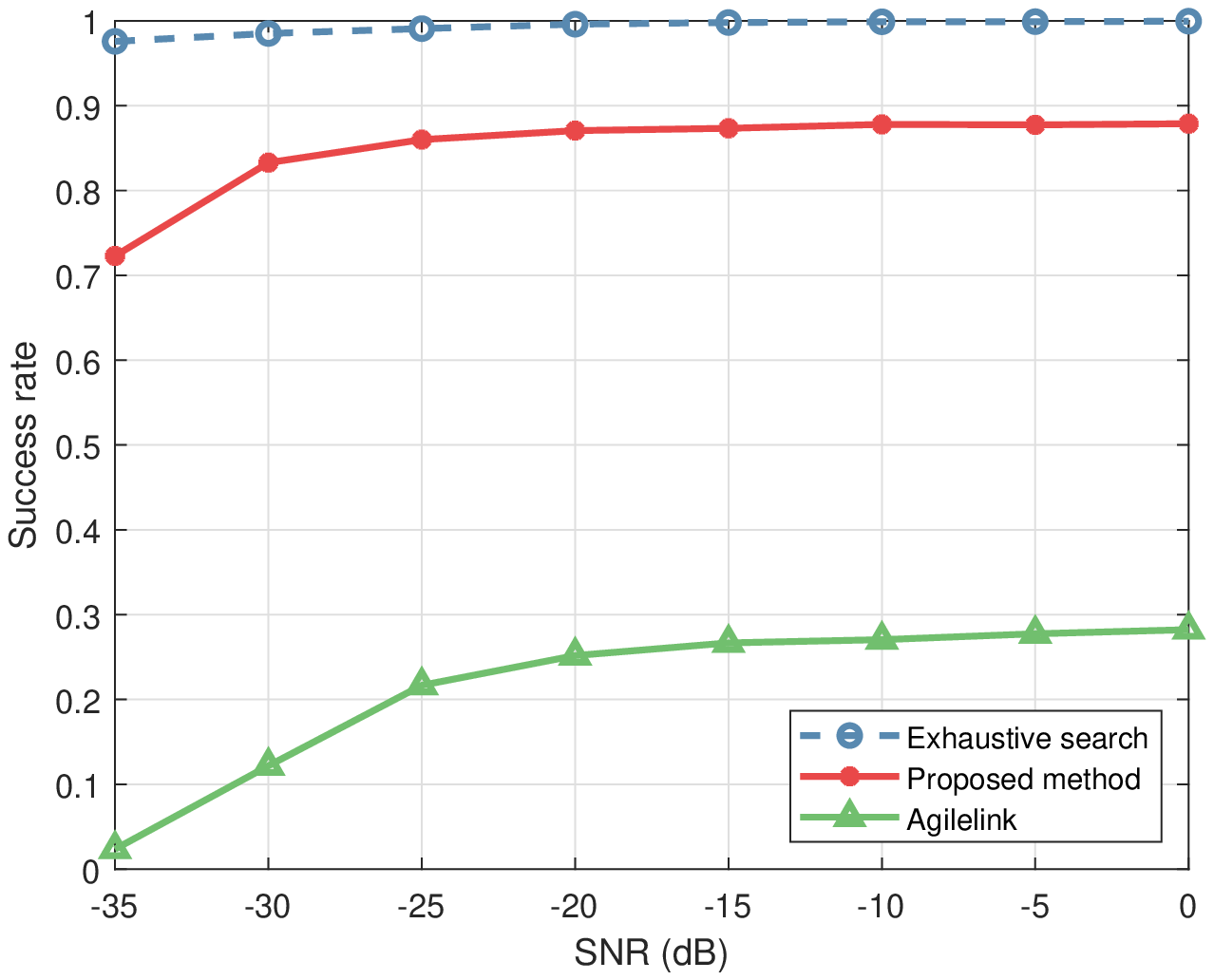}} \hfil
\subfigure[BGR versus SNR for LOS
scenarios.]{\includegraphics[width=2.8in]{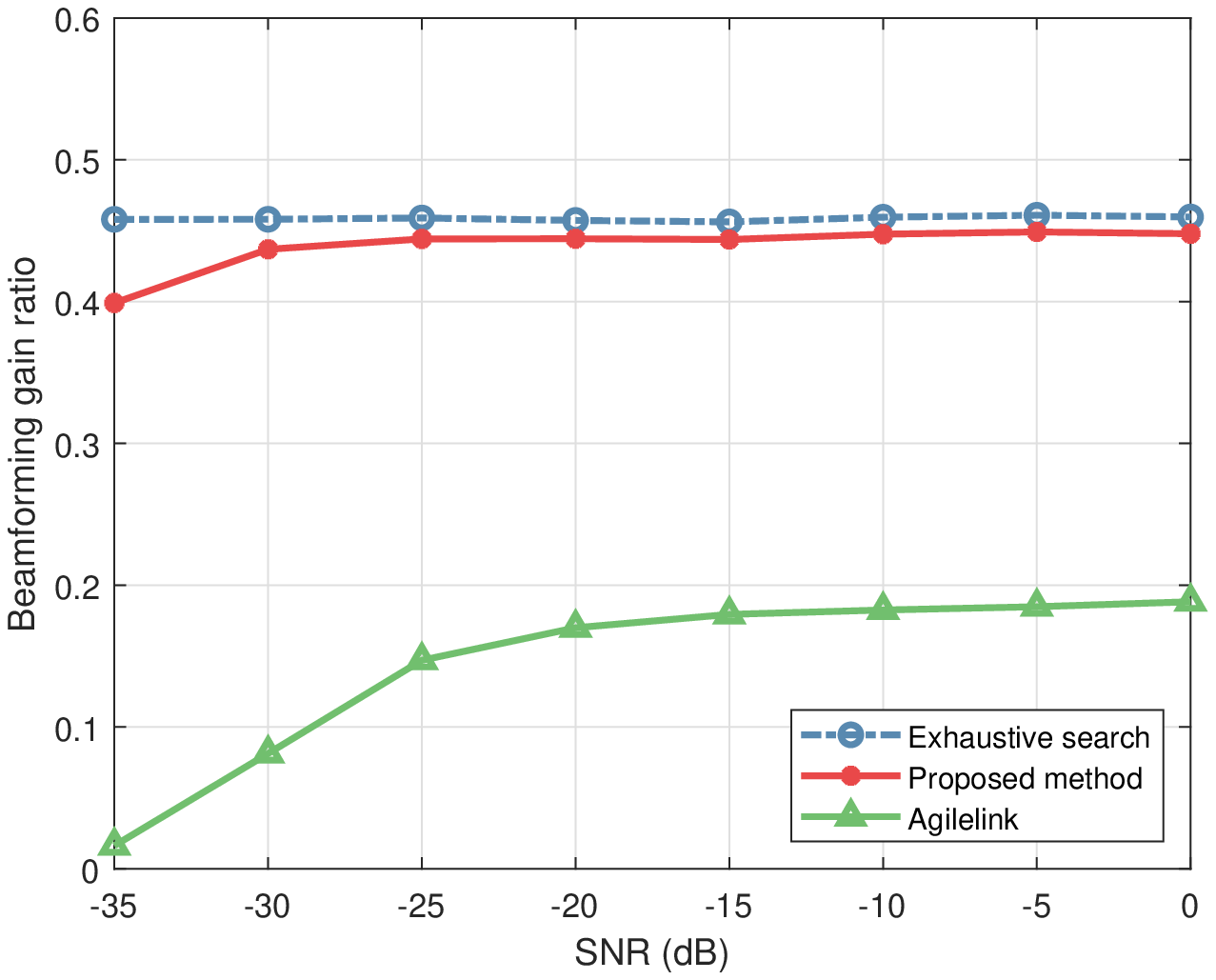}}\hfil
 \subfigure[Success rate versus SNR for NLOS scenarios.]
 {\includegraphics[width=2.8in]{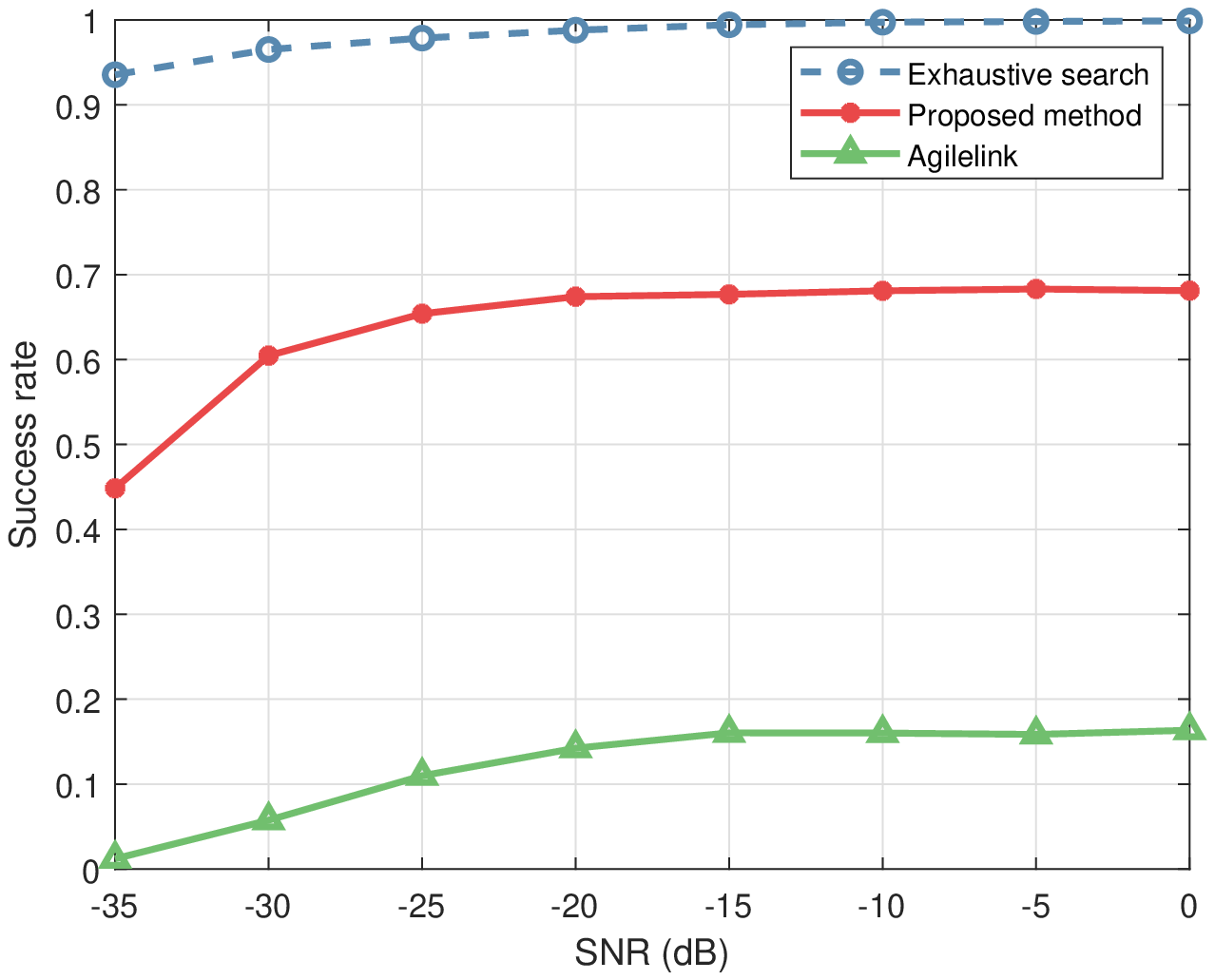}} \hfil
\subfigure[BGR versus SNR for NLOS
scenarios.]{\includegraphics[width=2.8in]{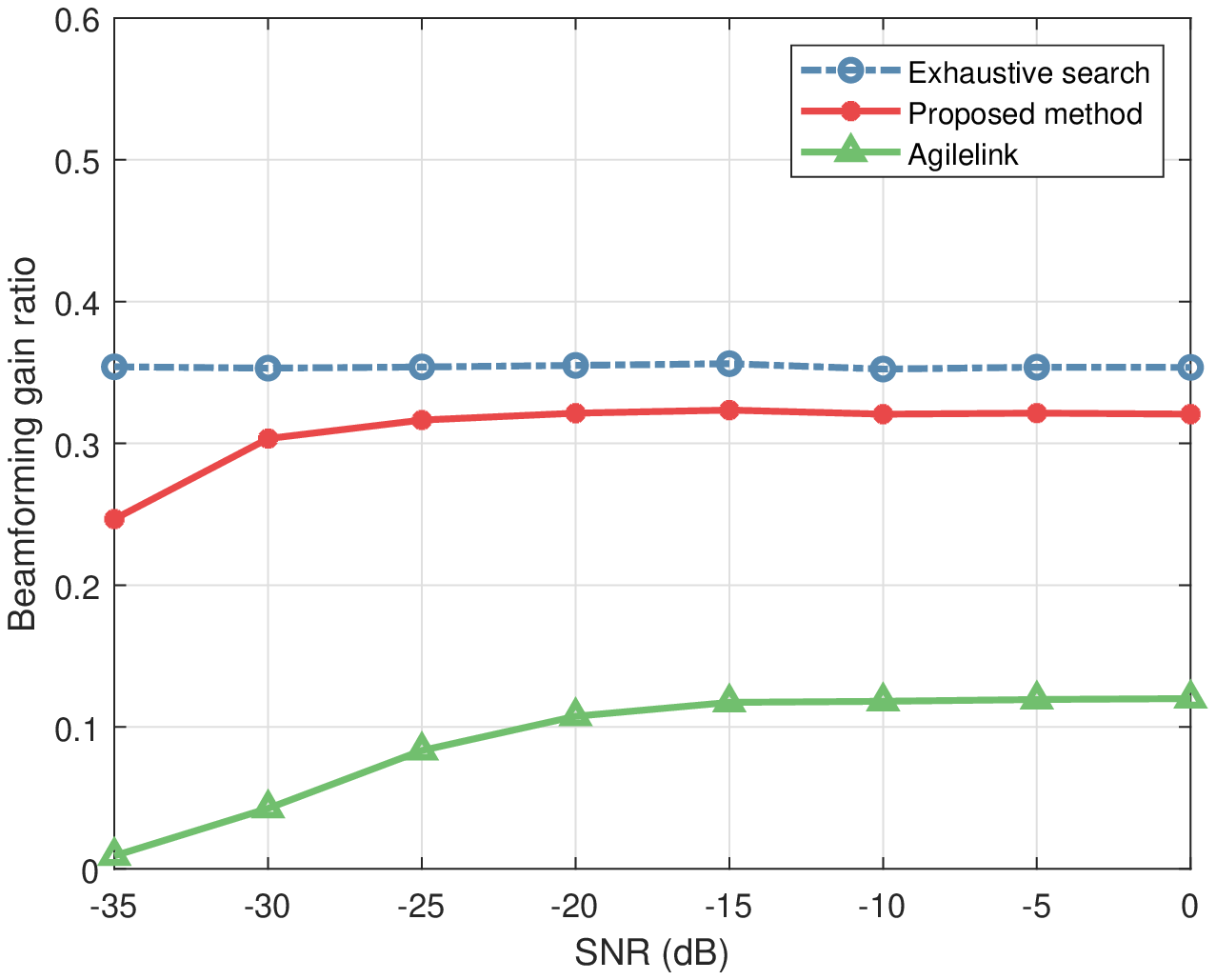}}
\caption{Success rate and beamforming gain ratio versus SNR for
LOS and NLOS scenarios in the noisy case, $T =1024$.} \label{fig_SNR_large}
\end{figure*}

\section{Simulation results} \label{sec:simulation}
In this section, we provide simulation results to illustrate the
performance of our proposed method. \textcolor{black}{In our simulations, we
consider a mmWave system operating at a carrier frequency of
28GHz. For the large-scale path loss, the reference channel power
gain at a distance of $1$m is set as $g_0 = -61.3$dB, the path
loss exponents of the BS-IRS and IRS-user links are set as
$\gamma_{BI} = 2.3$ and $\gamma_{IU} = 2$ respectively \cite{YouZheng20}. The
small-scale fading is modeled by Rician fading, with the BS-IRS
Rician factors set as $13.2$dB and the IRS-user Rician factors set
as $13.2$dB for LOS scenarios and $0$dB for NLOS scenarios. Also,
we assume the IRS adopts a UPA with $M = M_y \times M_z = 16
\times 16 = 256$ elements, and the BS employs a ULA with $N_t =
128$ antennas.} Two metrics are used to evaluate the performance of
the proposed method, namely, the success rate which is computed as
the ratio of finding the correct index of the largest (in
magnitude) component in $\boldsymbol{\Lambda}$, and the
beamforming gain ratio (BGR) which is defined as
\begin{align}
\gamma_{\text{bf}}\triangleq \mathbb E \left[  \frac{ \|
\boldsymbol{v}^H\boldsymbol{H} \boldsymbol{f}
 \|^2}{ \| \boldsymbol{v}_{\rm opt} ^H\boldsymbol{H} \boldsymbol{f}_{\rm opt} \|^2  } \right]
\end{align}
where $\boldsymbol{v}_{\rm opt}$ and $\boldsymbol{f}_{\rm opt}$
are optimally devised by assuming the full knowledge of the
cascade channel $\boldsymbol{H}$ via the method developed in
\cite{YuXu19}, $\boldsymbol{v}$ and $\boldsymbol{f}$ are devised
to align the BS's and IRS's beams to the strongest path, i.e.,
$\boldsymbol{v} =\sqrt{M} \boldsymbol{\bar D}_R (:,i^{\ast}) $,
$\boldsymbol{f} = \boldsymbol{D}_{N_t} (:,j^{\ast})$, and
$(i^{\ast}, j^{\ast})$ denotes the estimated index of the largest
component in $\boldsymbol{\Lambda}$. For a fair comparison, the
transmit beamforming vectors $\boldsymbol{f}(t)$ used by our
proposed method and other competing beam alignment schemes are
normalized to unit norm throughout our simulations. The
signal-to-noise ratio (SNR) is defined as $10 \log_{10}(\|
\boldsymbol{H}\|_F^2/(N_tM\sigma^2)$.

We now examine the performance of our proposed method and compare
it with the exhaustive beam search scheme discussed in Section
\ref{sec:system}.C and the state-of-the-art AgileLink scheme
\cite{HassaniehAbari28}. AgileLink is a beam alignment scheme
which also relies on the magnitude of measurements for recovery of
signal directions. Although originally developed for conventional
mmWave systems, its variant for array transmitter and receiver
(see Section 4.4 of \cite{HassaniehAbari28}) can be readily
applied to IRS-assisted mmWave systems since the IRS-assisted
signal model can be thought of as a conventional MIMO model.
AgileLink divides both the BS antennas and IRS elements into
several subarrays to form hashing beam patterns. For our proposed
method, the sparse encoding matrices $\{\boldsymbol{C}_l\}$ are
generated via the optimization-based method (\ref{v-opt-robust})
such that the corresponding passive beamforming vectors are of
constant modulus. In our simulations, we set $P=2$, $P'=2$, and
the AoAs and AoDs are randomly generated without assuming lying on
the discretized grid. For the LOS scenario, we assume that the
Rician factors for both the BS-IRS channel \eqref{ch-G} and the
IRS-user channel \eqref{hr} are set to $13.2$dB
\cite{AyachRajagopal14,AkdenizLiu14}. For the NLOS scenario, the
Rician factor for the IRS-user channel is set to $0$dB to simulate
the scenario where there are multiple comparable paths from the
IRS to the user.



In Fig. \ref{fig_L_large}, we plot the success rates and
beamforming gain ratios of respective methods as a function of the
total number of measurements $T$, where the SNR is set to
$\text{SNR} = -20$dB. For our proposed method, we set $R=8$ and
$Q=16$ and vary $L$ from $2$ to $8$, while for the AgileLink, for
each value of $T$, its parameters are carefully adjusted to
achieve its best performance. The exhaustive  search scheme is
included to provide the best achievable performance for any beam
alignment schemes, but it requires as many as $T=MN_t=32768$
measurements in total. From Fig. \ref{fig_L_large}, it can be seen
that the proposed method achieves a significant performance
improvement over the AgileLink scheme for both the LOS and NLOS
scenarios. Also, we observe that our proposed method, with only a
mild number of measurements (say, $T=1792$), can achieve a
beamforming gain close to the exhaustive search scheme.
Specifically, to achieve performance similar to the exhaustive
search scheme, the training overhead required by our proposed
method is only about 5\% of that needed by the exhaustive search
scheme.

In Fig. \ref{fig_SNR_large}, we plot the success rates and
beamforming gain ratios of respective methods as a function of the
SNR, where the number of measurements is set to $T=UVL = 16 \times
16 \times 4=1024$. From Fig. \ref{fig_SNR_large}, we can see that
our proposed method performs well even when the SNR is as low as
$-35$dB. Also, the performance of our proposed method is close to
that of the exhaustive search scheme while the proposed method
requires only $1024$ measurements, thus enjoying a substantial
training overhead reduction compared to the exhaustive search
scheme of $T=MN_t = 32768$. Moreover, the proposed method
outperforms the AgileLink scheme by a big margin across different
SNR regimes.

To examine the impact of the number of reflecting elements on the
proposed beam training method, in Fig. \ref{fig_M} , we plot the
success rates and beamforming gain ratios of respective algorithms
versus the number of reflecting elements for LOS scenarios, where
the SNR is set to $-20$dB and the total number of measurements
used for training (except the exhaustive search scheme) is set to
$T= 1024$. For our proposed method, we set $U=16$, $V=16$, and
$L=4$. Note that to make $U=M/Q$ unchanged, we adjust the value of
$Q$ accordingly for different choices of $M$. It can be observed
that the success rate of the proposed method keeps almost
unaltered as the number of reflecting elements $M$ increases,
whereas the AgileLink incurs a certain amount of performance loss
as $M$ grows.





Finally, to compare with SwiftLink, a fast compressed
sensing-based beam alignment algorithm that is robust against the
CFO \cite{MyersMezghani18}, we consider a simplified scenario
where the BS has the knowledge of the location of the IRS and has
aligned its beam to the LOS component between the BS and the IRS.
In this case, we only need to focus on beam training between the
IRS and the user. Our proposed method can be readily applied to
this scenario. Unlike AgileLink, SwiftLink cannot be
straightforwardly extended to joint BS-IRS-user beam training. But
it can be directly applied to this simplified scenario by treating
the IRS as an active transmitter. In our simulations, we set
$\rm{CFO} = 20$ppm, the carrier frequency is set to $28$GHz, the
bandwidth is set to $30$MHz and the phase noise is set to zero.
Also, due to the limitation inherent in the trajectory design,
SwiftLink only allows a limited number of measurements for
training, i.e. $T \in \{ T \leq 4M_y-2, \text{and }T = 4x, x
\text{ is an integer.} \}$. To satisfy this condition, we set
$T=60$ for SwiftLink. For a fair comparison, the number of
measurements for Agilelink and our proposed method is set to
$T=64$. Fig. \ref{fig_SNR_1D} depicts success rates and
beamforming gain ratios of respective algorithms. It can be
observed that our proposed method presents a substantial
performance improvement over Swiftlink, particularly in the low
SNR regime. \textcolor{black}{Since SwiftLink comprises several
sequential stages, the estimation accuracy of both the CFO and the
channel depends on the estimation results obtained in the previous
stages. Also, it is known that compressed sensing algorithms tend
to be fragile in low-SNR scenarios. In contrast, our proposed
method uses multi-directional beams to probe the channel and
relies on prominent measurements to identify the beam directions,
and thus is more resilient to low SNRs. This is probably the
reason why our proposed method outperforms SwiftLink, particularly
in the low SNR regime.}








\begin{figure*}[!t]
\centering
 \subfigure[Success rate versus $M$ for LOS scenarios]
 {\includegraphics[width=2.8in]{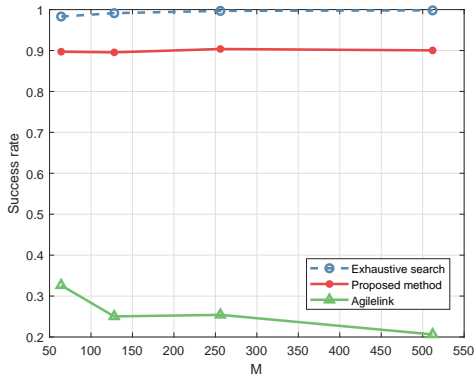}} \hfil
\subfigure[BGR versus $M$ for LOS
scenarios.]{\includegraphics[width=2.8in]{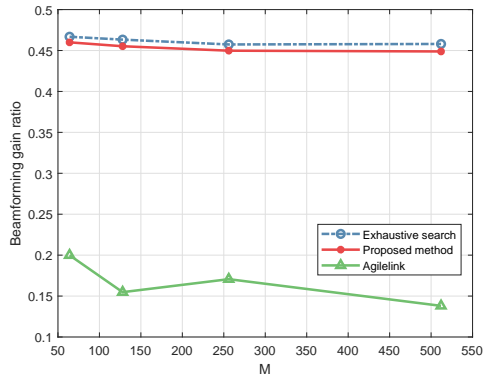}}\hfil
\caption{Success rate and beamforming gain ratio versus the
number of reflecting elements $M$ for LOS scenarios. } \label{fig_M}
\end{figure*}

\begin{figure*}[!t]
\centering
 \subfigure[Success rate versus SNR for LOS scenarios]
 {\includegraphics[width=2.8in]{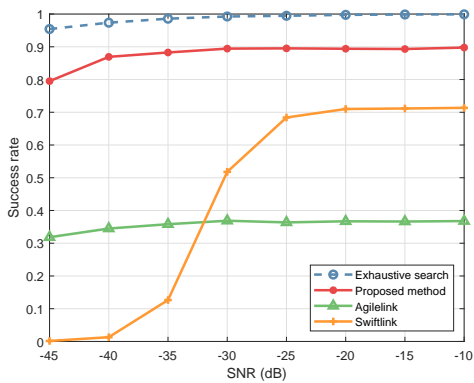}} \hfil
\subfigure[BGR versus SNR for LOS
scenarios.]{\includegraphics[width=2.8in]{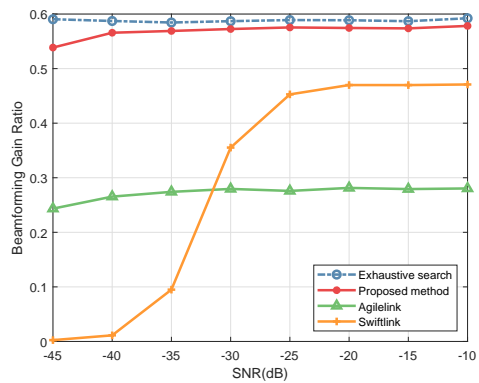}}\hfil
\caption{Success rate and beamforming gain ratio versus SNR for
LOS scenarios when the BS-IRS channel is perfectly aligned.} \label{fig_SNR_1D}
\end{figure*}

\section{Conclusions}
\label{sec:conclusion} In this paper, we studied the problem of
beam alignment for IRS-assisted mmWave/THz downlink systems. By
exploiting the inherent sparse structure of the BS-IRS-user
cascade channel, we devised multi-directional beam training
sequences to scan the angular space and proposed an efficient
set-intersection-based scheme to identify the best beam alignment
from compressive phaseless measurements. Theoretical and numerical
results show that the proposed method can perform reliable beam
alignment in the low SNR regime with a substantially reduced beam
training overhead.



\useRomanappendicesfalse
\appendices
\section{Proof of Theorem \ref{theorem1}}
\label{appendixA} We first define
\begin{align}
\mathcal{S} \triangleq \left \{ (i,j) | i \in
\bigcap_{l=1}^L\mathcal{S}_{u_l}^{(l)} , j \in
\bigcap_{l=1}^L\mathcal{S}_{v_l}^{(l)}  \right \}
\end{align}
From \eqref{opt-index}, we have $ (i^{\ast} ,j^{\ast} )\in
\mathcal{S} $. Let $B_{x} \triangleq \{ |\mathcal{S}| = x\}$
denote the event that the set $\mathcal{S}$ contains $RQ\geq x\geq
1$ elements in total, and $Z$ denote the event of identifying the
location of the largest element in $\boldsymbol{\lambda}$. We
therefore have
\begin{align}
\Pr(Z) = \sum_{x=1}^{RQ} \Pr(Z|B_x)\Pr(B_x)
\end{align}
where
\begin{align}
\Pr(Z|B_x) = \frac{1}{x}
\end{align}
Clearly, we have $\Pr(Z) \geq\Pr(B_1)$.
For simplicity, we only analyze the probability that the
intersection set $\mathcal{S}$ contains only one element.


Let $\mathcal{C}_l \triangleq \mathcal{S}_{u_l}^{(l)} - \{i^{\ast}
\}$ and $\mathcal{A}_l \triangleq\mathcal{S}_{v_l}^{(l)}
-\{j^{\ast} \}$. Define $\mathcal{S}_c \triangleq \bigcap_{l=1}^L
\mathcal{C}_l $ as the intersection set of the row indices, and
$\mathcal{S}_a \triangleq \bigcap_{l=1}^L \mathcal{A}_l $ as the
intersection set of column indices. Then we have
\begin{align}
\Pr(B_1) = \Pr( \mathcal{S}_c = \emptyset)  \times \Pr(
\mathcal{S}_a = \emptyset)
\end{align}
Since the two events $\mathcal{S}_c = \emptyset$ and
$\mathcal{S}_a = \emptyset$ are mutually independent, we can first
calculate $\Pr(\mathcal{S}_c = \emptyset)$ and the latter can be
obtained similarly.

Note that
\begin{align}
\Pr( \mathcal{S}_c = \emptyset)=1-\Pr( \mathcal{S}_c\neq\emptyset)
\end{align}
we now derive the probability of the event that the intersection
set $\mathcal{S}_c$ is non-empty. Without loss of generality, we
assume $i^{\ast}=M$. Let $D_m$ denote the event of
$m\in\mathcal{S}_c$. It can be easily verified that
\begin{align}
&\Pr \left(  \bigcap_{m \in J} D_m\right)  =
\bigg(\frac{\binom{M-1-j}{Q-1-j}}{\binom{M-1}{Q-1}} \bigg)^L , \nonumber \\
& J \subset \{1,2,\ldots,M-1 \}, |J| = j,j =1, \ldots,Q-1
\end{align}

The probability of the set $\mathcal{S}_c$ being non-empty can be
calculated as
\begin{align}
& \Pr(\mathcal{S}_c\neq \emptyset )
= \Pr \left(\bigcup_{m=1}^{M-1} D_m \right) \nonumber \\
\stackrel{(a)}=& \sum_{j=1}^{M-1}  (-1)^{j-1} \sum_{J \subset \{1,2,\ldots,M-1 \},
|J| = j}\Pr \left(  \bigcap_{m \in J} D_m\right) \nonumber \\
\stackrel{(b)}=&  \sum_{j=1}^{Q-1}  (-1)^{j-1} \sum_{J \subset \{1,2,\ldots,M-1 \},
|J| = j}\Pr \left(  \bigcap_{m \in J} D_m\right) \nonumber \\
\stackrel{(c)}= & \sum_{j=1}^{Q-1} (-1)^{(j-1)} \binom{M-1}{j}
\bigg(\frac{\binom{M-1-j}{Q-1-j}}{\binom{M-1}{Q-1}} \bigg)^L  \nonumber \\
\stackrel{(d)}\leq & (M-1) \left(\frac{Q-1}{M-1} \right)^L
\label{Sc-nonempty}
\end{align}

where $(a)$ follows from the principle of inclusion-exclusion, in
$(b)$, we utilize the property that
\begin{align}
\Pr \left(  \bigcap_{m \in J} D_m\right) = 0, \quad \forall |J|
\geq Q
\end{align}
since the intersection set $\mathcal{S}_c$ contains at most $Q-1$
elements, $(c)$ comes from the fact that there are
$\binom{M-1}{j}$ ways to select $j$ elements from the set $\{1,2,
\ldots,M-1 \}$ to make $|J|=j$, and $(d)$ follows from
\begin{align}
\Pr \left(\bigcup_{m=1}^{M-1} D_m \right) \leq  \sum_{m=1}^{M-1}
\Pr(D_m) = (M-1) \Pr(D_m)
\end{align}

As a result, we have
\begin{align}
&  \Pr(\mathcal{S}_c  =\emptyset)
= 1-\Pr(\mathcal{S}_c\neq \emptyset )\nonumber \\
\stackrel{(a)}\geq &  1- (M-1) \left(\frac{Q-1}{M-1} \right)^L \nonumber \\
= & P(Q,L,M) \label{P-QL}
\end{align}
where $(a)$ is derived from \eqref{Sc-nonempty}.



Similarly, we have
\begin{align}
\Pr(\mathcal{S}_a = \emptyset) \geq  1- (N_t-1)
\left(\frac{R-1}{N_t-1} \right)^L = P(R,L,N_t)
\end{align}
Therefore the probability of identifying the location of the
largest component in $\boldsymbol{\Lambda}$ is no smaller than
\begin{align}
{\rm P} = \Pr(Z) \geq P(Q,L,M) \times P(R,L,N_t)
\end{align}
This completes our proof.

\section{Proof of Theorem \ref{theorem2}}
\label{appC} Let $X$ denote the number of NM rounds out of the
total $L$ full-coverage rounds of scanning, and $\tilde Z$ denote
the event of exact recovery of the location of the largest
element. Specifically, define a random variable
\begin{equation}
    X_l \triangleq \begin{cases}
    1, & \text{if round $ l$ is a NM round } \\
    0, & \text{if round $ l$ is not a NM round }
           \end{cases}
\end{equation}
Thus $X$ can be expressed as
\begin{align}
X \triangleq \sum_{l=1}^{L} X_l. \label{X-binom}
\end{align}
Clearly, the random variables $\{X_l\}$ are mutually independent
and identically distributed with $\Pr (X_l = 1) =p,\Pr(X_l = 0) =
1-p$. The random variable $X$, therefore, follows a binomial
distribution, i.e., $X \sim {\text B}(L,p)$. Thus we have
\begin{align}
\Pr(\tilde Z) = &  \sum_{l = 0}^{L} \Pr(\tilde Z|X=l) \Pr(X=l) \nonumber \\
\stackrel{(a)}=&  \sum_{l = 0}^{L} \Pr(\tilde Z|X=l)  \binom{L}{l} p^{l}(1-p)^{L-l} \nonumber \\
\stackrel{(b)}\geq & \sum_{l = 0}^{L}  g(Q,l,M) \times g(R,l,N_t)
\times \left( \binom{L}{l} p^{l}(1-p)^{L-l}  \right)
\label{prob-NLOS}
\end{align}

where $(a)$ is obtained from the probability mass function of the
binomial distribution; and in $(b)$, we directly apply the
equality $(c)$ of \eqref{Sc-nonempty}.

\bibliography{newbib}
\bibliographystyle{IEEEtran}

\end{document}